\documentclass[12pt]{article}

\usepackage[T1]{fontenc}

\usepackage{amsmath,amssymb,amsthm,amsfonts}

\usepackage{fullpage,multirow,array,cite,bbm,bigints}

\usepackage[ruled]{algorithm2e}

\usepackage{adfbullets}

\usepackage{subcaption}

\usepackage{tikz} 
\usetikzlibrary{positioning}
\usetikzlibrary{decorations.pathreplacing}

\usepackage{graphicx}
\graphicspath{ {./images/} }

\theoremstyle{plain} 
\newtheorem{theorem}{Theorem}[section]
\newtheorem{lem}[theorem]{Lemma}

\newtheorem{cor}[theorem]{Corollary}
\newtheorem{algo}[theorem]{Algorithm}

\theoremstyle{definition}
\newtheorem{defn}{Definition}[section]

\newtheorem{exmp}{Example}[section]

\theoremstyle{remark}
\newtheorem*{rem}{Remark}

\newcommand{\defeq}{\mathrel{\mathop:}= }

\renewcommand{\H}{\mathcal{H}}
\newcommand{\I}{\mathcal{I}}
\newcommand{\J}{\mathcal{J}}
\newcommand{\K}{\mathcal{K}}

\newcommand{\A}{\mathcal{A}}
\newcommand{\SSS}{\mathcal{S}}

\newcommand{\G}{\mathcal{G}}

\newcommand{\RRR}{\mathcal{R}}
\newcommand{\U}{\mathbf{U}}
\newcommand{\X}{\mathcal{X}}
\newcommand{\Y}{\mathcal{Y}}

\newcommand{\E}{\mathbb{E}}
\renewcommand{\P}{\mathbb{P}}

\newcommand{\R}{\mathbb{R}}

\newcommand{\dd}{\mathbf{d}}
\newcommand{\p}{\mathbf{p}}
\newcommand{\bbb}{\mathbf{b}}
\newcommand{\aaa}{\mathbf{a}}

\newcommand{\s}{\mathbf{s}}

\newcommand{\h}{\mathbf{h}}
\newcommand{\ppp}{\mathbf{p}}
\newcommand{\qqq}{\mathbf{q}}
\newcommand{\q}{\mathbf{q}}
\newcommand{\x}{\mathbf{x}}
\newcommand{\y}{\mathbf{y}}

\newcommand{\ttt}{\mathbf{t}}

\DeclareMathOperator*{\argmax}{arg\,max}
\DeclareMathOperator*{\argmin}{arg\,min}

\begin{document}
\title{Autocratic Strategies of Multi-State Games}
\date{December 28, 2020} 
\author{Mario Palasciano}

\maketitle


\begin{abstract}
In a single-state repeated game, zero-determinant strategies can unilaterally force functions of the payoffs to take values in particular closed intervals. When the explicit use of a determinant is absent from the analysis, they are instead called autocratic. We extend their study to the setting of finite state games with deterministic transitions. For a given game we find that the endpoints of the intervals of enforceable values must satisfy fixed point equations. From these extreme enforceable values we show it is always possible to construct finite memory strategies to enforce a particular value. An algorithm is presented which will approximate the enforceable values in each state. Finally, we present formulas from which the exact solutions can be recovered from the approximate solution.
\end{abstract}

\section{Introduction}

Game theory provides a framework for studying interactions between agents competing over resources, typically characterized by a scalar utility function. Among the most compelling games are social dilemmas, which result from conflicts of interest between individuals and groups. The classical model of a social dilemma is the Prisoner's Dilemma, a two-player game with the actions ``cooperate" and ``defect''. Rational players are incentivized to defect regardless of their opponents action, leading to lower individual payoffs than if they had displayed mutual cooperation. 

It turns out this suboptimal outcome is not an empirical certainty among human players (Heide and Miner \cite{HEIDE1992} and Bo \cite{DALBO2005}). One proposed explanation for this cooperative tendency is that human beings instinctively treat one-shot games as the first of a random, or at least unknown, number of one-shot games. The repeated game encourages cooperation by allowing for reciprocation of earlier behaviour. Perhaps the simplest example of such a reciprocating strategy for the repeated Prisoner's Dilemma is one which copies their opponent's action in the previous round. Known as ``tit-for-tat", this strategy was very successful in tournaments involving a diverse population of strategies (Trivers \cite{TRIVERS1971}).

The observed altruistic impulse also helped prompt the exploration of games played amongst a population that evolves or learns over time. Evolutionary game theory is the study of these dynamical systems. In this setting, a strategy's reproductive potential is mediated by its performance against a population of other strategies, including its own. One proposed route to cooperation in the evolutionary game is through kin selection, in which players engage in self-sacrifice to ensure the genetic success of their relatives \cite{MARSHALL}.

In their 2012 paper \cite{PD2012}, Press and Dyson explictly construct a class of memory-$1$ strategies for the repeated Prisoner's Dilemma which share a curious property. These strategies, dubbed zero-determinant strategies, unilaterally enforce a linear relation between the players' expected payoffs, offering a single player a level of control much greater than previously conjectured. In particular, if one's opponent evolves by way of gradient learning methods, they demonstrate that these strategies can lead to a greater share of the total payoff in future generations. 

Since their inception, these strategies has been extended to cover multiplayer social dilemmas (\cite{HILBE2015}, \cite{HILBE2014} and \cite{GOVAERT2020}) and games of incomplete information(\cite{UEDA2018}, \cite{ZD_OBSERVATION_ERROR}). In addition, their robustness has been explored in the context of evolutionary game theory (\cite{ADAMI2013}, \cite{STEWART2013} and \cite{SZOLNOKI2014}). In each of these studies, the game is assumed to have only two actions in order to make use of the vanishing matrix determinant argument employed by Press and Dyson. In \cite{MH2016}, Hauert and McAvoy generalize these strategies by considering games with arbitrary action spaces. They dub this new class the \textbf{autocratic strategies}. The term autocrat is perhaps best suited for asymmetric games, studied further in \cite{REPEATEDASYMMETRIC}.

Of more recent interest is the impact of multiple game states on the evolution of cooperation. For example, cooperative play in which cooperation allows access to subgames with greater rewards can greatly enhance the propensity for cooperation, as in \cite{Su25398}. There, Su et al present the example of a two state Prisoner's Dilemma game. Mutual cooperation one round lets the players vie for a higher reward the next round. Towards this end we investigate enforceable values of the expected total utility in multi-state games with deterministic transitions. 

Zero-determinant strategies have found their way into applications including crowdsourcing quality control \cite{CROWDSOURCE} and computing delegation \cite{DELEGATE}, and the cyberdefense of electrical grids \cite{PRICEATTACK}, the internet of things \cite{INTERNETOFTHINGS} and blockchain currency \cite{BLOCKATTACK}.

\subsection{Description of the Results}

Some notation is required to state our results. The game is played for a random number of rounds on a directed graph whose nodes represent the various states $\SSS$ of the game. While in state $s \in \SSS$, the autocrat and their opponent choose their state-specific actions from their respective action $\X_s$ and $\Y_s$, respectively. The outgoing edges of node $s$ represent the action pairs drawn from this joint action space, denoted $\A_s$. The target node of outgoing edge $(x,y)$ in state $s$ is given by a transition function $T_s:\A_s \rightarrow \SSS$. Each joint action $a=(x,y)$ in state $s$ is associated with a scalar utility $U(a;s)$, typically taken in the literature to be a linear function of the players' individual classical payoff functions. 

The game begins in node $s_0$ in round $0$. The players choose their actions $x_0$ and $y_0$, whereupon they travel along the edge corresponding to said action and accrue a utility of $U(x_0,y_0;s_0)$. A biased coin with parameter $\lambda$ is then consulted. If heads, the game continues on in this fashion for another round. Otherwise, the game ends, and the total utility is simply the sum of the utilities garnered during each round of play. As zero-determinant strategies were originally studied in the $\lambda \rightarrow 1^{-}$ limit, we will divide the total utility by the average number of rounds played, $(1-\lambda)^{-1}$, to ensure convergence. The rescaled utilities are denoted $U_{\lambda}(a;s)$. 

Suppose the players find themselves in state $s$ at the beginning of round $k$. The sequence of previous action-state pairs, or the history of play, is denoted $h_{<k}$. In general, the autocrat and opponent choose their actions by independently sampling the probability distributions $p_{ \cdot | s,h_{<k}}$ and $ q_{\cdot | s,h_{<k}} $ on $\X_s$ and $\Y_s$. The distribution on $\X \defeq \cup_s \X_s$ in round $k$ is denoted $p_k$. The collection $(p_k)_{k \geq 0}$, denoted $\p$, is known as the behavioural strategy of the autocrat. The $q_k$'s and $\q$ are defined analogously.

We investigate a restricted class of autocratic strategies for multi-state games analogous to those explored in the literature  covering single state games (\cite{MH2016}). These strategies imposed the condition that expected future utility does not depend on the opponent's current action. That is,
\begin{align}
\label{restricted_class}
\Phi_{k} ( x,y; s , h_{< k}) \equiv \Phi_{k} ( x; s , h_{< k}) \quad \forall \: k \geq 0.
\end{align}
\noindent We will call such strategies opponent agnostic. 

Our main result, Theorem \ref{THM_main_result}, says that opponent agnostic autocratic strategies need finite memory to fix the expected total utility. One might hope that it suffices to use a strategy which only uses memory-$1$ strategy, but sometimes the autocratic player is permitted to play only one action in a chain of states. In this case an autocratic strategy requires memory equal length of the longest of these chains. Example \ref{EXMP_agnostic_pathological} presents a family of such games.  

The construction of the autocratic strategies hinges on the characterization of values enforceable by opponent agnostic strategies in Theorem \ref{THM_characterization_opposition_agnostic}. Upon fixing the largest support of the autocrat's mixed strategy to be $X_s'$ in each state $s$, the potential minimum and maximum enforceable values of the expectation are the unique solutions to the fixed point equations
\begin{align}
\label{eqn_m_s}
m_{s} \defeq \min_{x \in \X_s'} \max_{y \in \Y_s} \big\{ \lambda m_{T_s(x,y)} + U_{\lambda} (x,y,s)  \big\}
\end{align}
and
\begin{align}
\label{eqn_M_s}
M_{s} \defeq \max_{x \in \X_s'} \min_{y \in \Y_s} \big\{ \lambda M_{T_s(x,y)} + U_{\lambda} (x,y,s)  \big\}.
\end{align}
The $\argmax_x$ and $\argmin_x$ of the above, that is
\begin{align} 
\label{extremal_left_description}
x^-_s & \defeq \argmin_{x \in \X'_s} \max_{y \in \Y_s} \big\{ \lambda m_{T (x,y;s)} + U_{\lambda} (x,y;s) \big\}
\end{align}
and
\begin{align}
\label{extremal_right_description}
x^+_s & \defeq \argmax_{x \in \X'_s} \min_{y \in \Y_s} \big\{ \lambda M_{T (x,y;s)} + U_{\lambda} (x,y;s) \big\} 
\end{align}
are called the \textbf{extremal actions}. Theorem \ref{THM_main_result} says that the solutions to the above fixed point equations must satisfy 
\begin{align}
\label{inequality_description}
\max_{y \in \Y_s} \big\{ \lambda m_{T (\x_s,y;s)} + U_{\lambda} (x,y;s) \big\} \leq \min_{y \in \Y_s} \big\{ \lambda M_{T (\x_s,y;s)} + U_{\lambda} (x,y;s) \big\} 
\end{align}
for any extremal action $\x_s$ in order for $[m_s,M_s]$ to enforceable.

For a fixed family of autocratic actions $(\X_s' : s \in \SSS)$, the solutions to (\ref{eqn_m_s}) and (\ref{eqn_M_s}) can be found approximately using a globally convergent iterative process. From the extremal actions  can exact formulas for $m_s$ and $M_s$ be constructed. These are equations (\ref{EQN_exact_cycle}) and (\ref{EQN_exact_branch}). 

In addition to our theoretical results, we present Algorithm \ref{ALGO} that initially runs the iterative process on graph corresponding to the full action space $\X$ available to the autocrat. If the inequality (\ref{inequality_description}) does not hold for some states, the autocrat's action space are pruned of the offending actions before running the iterative process again. In this fashion is the algorithm zeroes in on largest family $\X^f$.
\subsection{Outline of the Proofs}

The proof of the main tool, Theorem \ref{THM_Phi_condition}, borrows the notion of future total utility functions $(\Phi_k)_{k \geq 0}$ used in Markov decision processes. That is, $ \Phi_{k} ( a; s , h_{< k})$ is the expected total utility from round $k$ onwards assuming that the current action-state pair is $(a,s)$ and the history of play is $h_{<k}$. They obey the recursive \textbf{Bellman-type} equation 
\begin{align*}
 \Phi_{k} ( x_{k},y_k; s_{k} , h_{< k}) = U ( x_{k},y_k;s_{k})  + \lambda \sum_{x \in \X_{s}} \sum_{y \in \Y_s}  \Phi_{k+1} \left( x,y ;s , h_{< k+1} \right) p_{x| s,h_{< k+1}} q_{y| s,h_{< k+1}} .
\end{align*}
Theorem \ref{THM_Phi_condition} states that a strategy $\p$ is autocratic iff there exists $\Phi_k$'s, dependent only on $\q$, that satisfy the simplified Bellman equation, 
\begin{align}
\label{eqn_outline_bellman}
\Phi_{k} ( x_{k},y_k; s_{k} , h_{< k}) = U_{\lambda} ( x_{k},y_k;s_{k})  + \lambda \sum_{x \in \X_{s}} \Phi_{k+1} \left( x,y ;s , h_{< k+1} \right) p_{x| s,h_{< k+1}} \quad \forall \: y \in Y_s.
\end{align}
With the dependence on $\q$ removed, the problem can be viewed as a control problem on a class of Markov decision processes. 

The main idea of Theorem \ref{THM_characterization_opposition_agnostic} is find the extremal values that $\Phi_k (x,y;s,h_{<k})$ can take on. The proof of Theorem \ref{THM_main_result} then expresses the corresponding autocratic strategy $\p$ in terms of these extremal values. Fixing a value of the expected total utility and then using equation (\ref{eqn_outline_bellman}) to construct consistent $\Phi_k$'s would be straightforward if not for the complication that the $\Phi_k$'s must be bounded. 

More specifically, each $\Phi_k$ can only take on values between $m_0 \defeq \min_{(a,s)} U(a;s)$ and $M_0 \defeq \max_{(a,s)} U(a;s)$. This motivates the following definition. The value $v$ is \textbf{enforceable} in state $s$ for the game starting in state $s$ if there exists a strategy $\p$ of the autocrat such for every sequence of opponent actions $(y_k)_{k \geq 0}$ we have that $\Phi_k$ remains inside $[m_0,M_0]$. 

The enforceable values in state $s$ are found as follows. The initial exterior estimate for each state is taken to be $\left[m_{s,0}, M_{s,0} \right] \defeq \left[ m_0,M_0 \right]$. A proposed value of the expectation $v$ must satisfy the first step of the recursive formula (\ref{eqn_outline_bellman}). That is,
\begin{align}
\label{EQN_outline_1}
v = \sum_{x \in \X_{s}} \Phi_{0} \left( x,y ;s  \right) p_{x| s} .
\end{align}
The potential value $v$ is discarded if, for any $\Phi_0$ and $p_{\cdot|s}$ satisfying the above, there exists a joint action $(x,y)$ and child state $s_c = T(x,y;s)$ of $s$ such that
\begin{align}
\label{EQN_outline_2}
\lambda^{-1} \left[\Phi_{0} \left( x,y ;s  \right) - U_{\lambda} (x,y;s) \right] \notin [m_{s_c,0},M_{s_c,0}].
\end{align}
In this way a smaller estimate $\left[ m_{s,1},M_{s,1} \right]$ of the enforceable values in each state $s$ is formed. Solving the general case is the subject of future work. For now we restrict ourselves with opponent agnostic strategies, which constrain the expected future utility to not depend on the opponent's current action $y$. That is, $\Phi_0 (x,y;s) \equiv \Phi_0 (x;s)$. Thus a potential $\Phi_0(x;s)$ must lie in 
\begin{align}
\label{EQN_outline_3}
\bigcap_{y \in \Y_s} \bigg( \lambda [m_{T(x,y;s),0},M_{T(x,y;s),0}] + U_{\lambda}(x,y;s) \bigg).
\end{align}
That is,
\begin{align*}
 \max_{y \in \Y_s} \left\{ \lambda m_{T(x,y;s),0} + U_{\lambda} (x,y;s) \right\} 
 \leq \Phi_0(x;s) \leq  \min_{y \in \Y_s} \left\{ \lambda M_{T(x,y;s),0}  + U_{\lambda} (x,y;s) \right\}.
\end{align*}

As a result, the extremal values of $v$ are better approximated by
\begin{align*}
m_{s,1} = \min_{x \in \X_s} \max_{y \in \Y_s} \left\{ \lambda m_{T(x,y;s),0} + U_{\lambda} (x,y;s) \right\}
\end{align*}
and 
\begin{align*}
M_{s,1} = \max_{x \in \X_s} \min_{y \in \Y_s} \left\{ \lambda M_{T(x,y;s),0}  + U_{\lambda} (x,y;s) \right\}.
\end{align*}
This procedure is then iterated. Taking the limit gives us the range $[m_s,M_s]$ enforceable values, as well as the fixed point
equations 
\begin{align}
\label{EQN_outline_m_FPE}
m_{s} = \min_{x \in \X_s} \max_{y \in \Y_s} \left\{ \lambda m_{T(x,y;s)} + U_{\lambda} (x,y;s) \right\}
\end{align}
and 
\begin{align}
\label{EQN_outline_M_FPE}
M_{s} = \max_{x \in \X_s} \min_{y \in \Y_s} \left\{ \lambda M_{T(x,y;s)}  + U_{\lambda} (x,y;s) \right\}.
\end{align}
Note the extremal actions (\ref{extremal_left_description}) and (\ref{extremal_right_description}) in each state do not necessarily satisfy the inequality
\begin{align}
\label{EQN_outline_inequality}
\max_{y \in \Y_s} \left\{ \lambda m_{T(x,y;s)} + U_{\lambda} (x,y;s) \right\} 
 \leq  \min_{y \in \Y_s} \left\{ \lambda M_{T(x,y;s)}  + U_{\lambda} (x,y;s) \right\},
\end{align}
In this case, the game graph must be pruned of edges corresponding to the offending extremal actions before the iterative process is ran again. 

A benefit of implementing the iterative process shown above in Algorithm \ref{ALGO} is that it converges globally to the unique solution of the fixed point equations (\ref{eqn_m_s}) and (\ref{eqn_M_s}) on the game graph $(\SSS, \X' \times \Y,U)$. Issues can crop when removing said state's incoming edges from the graph. If the edge $(x,y)$ leads to $s$, then the edges corresponding to each of $(x,y')$ must also be removed. If this results in a state $s'$ with no outgoing edges, they must too be removed in turn. These two pruning steps alternate until the resulting graph is stable.

The algorithm terminates when a subgraph $(S^f,X^f \times \Y,U)$ is found such that the inequality (\ref{EQN_outline_inequality}) holds for all extremal actions or returns the empty graph. In the former case, $[m_s,M_s]$ represents the values the autocratic player can unilaterally enforce through an opponent agnostic strategy.

\section{Game Theoretical Background and Notation}

Informally, a \textbf{game} is a description of potential strategic interactions between individual players. For each player, this description must specify their available information, their possible actions, and furthermore must provide a measure of their preference for each outcome. The preference of each player is usually encapsulated in the form of individual payoff functions, with a larger value typically being the more desirable. Although utility is often used synonymously with payoff in the literature, in what follows it is best thought of as a linear combination of the classical payoff functions. 

\subsection{Multi-State Games}

A \textbf{multi-state game} is played on a directed graph whose nodes represent the various states $\SSS$ of the game. While in state $s$, the autocrat and their opponent choose their state-specific actions from $\X_s$ and $\Y_s$, respectively. The outgoing edges of node $s$ represent the action pairs drawn from the joint action space, $\A_s \defeq \X_s \times \Y_s$. The totality of the joint actions, $\sqcup_s \A_s$, are written as $\A$, with similar definitions for $\X$ and $\Y$. The game graph is denoted by $\G(\SSS, \A )$. We assume that the knowledge of the entire graph and the current occupied node is available to the players. That is, the game is one of perfect information.

A fixed-length game is played for a finite number of rounds. The rounds are numbered starting from $0$. The general play of the game is illustrated as follows. Suppose the players begin round $k$ is state $s$. They independently sample probability distributions on $\X_s$ and $\Y_s$, called mixed strategies. The players resulting joint action $a$ determines both the utility $U(a;s)$ the players accrue and the state they begin round $k+1$ in.

A temporally discounted game is instead played for a random number of rounds, unknown to the players. This random game length is generated in the following manner. At the end of each round, a \textit{Bernoulli}($\lambda$) random variable is independently sampled to determine whether another round is played. If $1$ is the outcome, the game proceeds to the next round. Otherwise, the game ends. Equivalently, the total number of rounds played can be generated by a random variable $\ttt_{\lambda} \sim$\textit{Geometric}($\lambda$). We assume that the players have access to the value of the discount factor.

We show in \eqref{discount_motivation} that expected utility received in the $n$th round are reduced by a factor of $\lambda^n$. For this reason we will refer to the parameter $\lambda$ as the discount factor. A repeated game is simply a discounted game with one state. 

It will be convenient to collect the destination of each outgoing edge $(x,y)$ from node $s$ in the transition function $T(a ;s)$. The \textbf{parents} of state $s$ are all $s_p$ such that there exists $a \in \A_{s_p}$ with ${s=T(a;s_p)}$. The \textbf{children} of state $s$ are all $s_c$ such that there exists $a \in \A_s$ with ${s_c=T(a;s)}$.

Each round, players are free to play a mixed strategy conditioned on the history of play. The \textbf{history} at round $n$, denoted $h_{< n}$, is a sequence of past action-state tuples for $m<n$. That is,
\begin{eqnarray*}
h_{<n} \defeq  \Big( \left( a_{n-1} , s_{n-1} \right) , \ldots , \left( a_1 , s_1 \right), \left( a_0 ,s_0 \right)  \Big), 
\end{eqnarray*}
where $s_{i+1} = T \left( a_i ;s_i \right)$. In addition, it will be useful to introduce the notation
\begin{eqnarray*}
h_{\leq n} \defeq  \Big( \left( a_{n} , s_{n} \right) , \ldots , \left( a_1 , s_1 \right), \left( a_0 ,s_0 \right)  \Big) 
\end{eqnarray*}
and
\begin{align*}
h_{[j,i]} = \Big( \left( a_{i} , s_{i} \right) , \left( a_{i-1} , s_{i-1} \right)  \ldots , \left( a_{j+1} , s_{j+1} \right), \left( a_j ,s_j \right)  \Big)  
\end{align*}
for the purpose of making expressions more compact. 

The set of all n-histories is denoted by $\H_n$. In particular, $H_0 \defeq \{ h_{<0} \}$ represents the singleton set consisting of the null history. The null history serves to indicate that there has been no history of play in round $0$. Furthermore, we let $\H$ denote the disjoint union of all $\H_n$. That is,
\begin{align*}
\H \defeq \bigsqcup_{n \geq 0} \H_n .
\end{align*} 

A player's history-dependent mixed strategies form a behavioural strategy. Formally, the autocrat chooses conditional probabilities $p_{x | s, h_{<n}}$ for all $x \in \X_s$ and $h_{<n} \in \H_n$. These are the probabilities they play $x$ in round $n$ given that the $n$-history $h_{<n}$ has deposited them in state $s$. We let $\ppp_n$ denote player $1$'s mixed strategies in round $n$. That is, 
\begin{align*}
\ppp_n \defeq \big( p_{x | s, h_{<n}} : x \in \X_s, h_{<n} \in \H_n \big) .
\end{align*}
In addition, $\ppp \defeq (\p_n)_{n \geq 0}$ denotes their entire behavioural strategy. Similar definitions are used for the opposing player's strategies $\qqq_n$ and $\qqq$.

Equivalently, the players' actions can be generated by random variables. We let $\x_n$, $\y_n$, and $\s_n$  denote the random variables which take values in $\X$ and $\Y$ and $\SSS$ in round $n$. The random vector $(\x_n,\y_n)$ is denoted by $\aaa_n$. 
We also define the history-valued random variables  
\begin{align*}
\h_{<n} \defeq  \Big( \left( \aaa_{n-1} ,\s_{n-1} \right) , \ldots , \left( \aaa_{0},\s_0 \right) \Big) \quad \mbox{and} \quad \h_{\leq n} \defeq  \Big( \left( \aaa_{n} ,\s_{n} \right) , \ldots , \left( \aaa_{0},\s_0 \right) \Big)  .
\end{align*}
 The conditional laws of $\x_n$ are given by
\begin{align*}
\P \Big( \x = x \Big| \s_n =s, \h_{<n} = h_{<n} \Big) \defeq p_{ x_n |s, h_{<n}}  \quad \forall \: x \in \X , s \in \SSS , h_{<n} \in \H_n .
\end{align*} 
As we assume the players select their actions independently of each other, the conditional law of the random vector $\aaa_n $ is given by 
\begin{align*}
\P \Big( \aaa_n = (x,y) \Big| \s =s, \h_{<n}  = h_{<n} \Big) \defeq p_{ x | s, h_{<n}} q_{ y | s, h_{<n}} \quad \forall \: (x,y) \in \A_s, s \in \SSS, h_{<n} \in \H_n. 
\end{align*}
For brevity we will write the conditional probability as
\begin{align}
\label{cond_prob_succinct}
P_{i+1} \big[ x \big| a_i,s_i,a_{i-1},s_{i-1} \ldots a_{j},s_{j} \big] \defeq \P \big( \x_{i+1} = x  \big| \h_{[j,i]} =  (a_i,s_i, \ldots a_{j},s_{j})   \big).
\end{align}
Note Kolmogorov's Extension Theorem guarantees the existence of a probability space on which both $(\aaa_n)_{n \geq 0}$ and $\ttt_{\lambda}$ are defined and independent. 

We define the random variable $\U$ to be the \textbf{total utility} garnered during the course of the game. That is,
\begin{align}
\label{U}
\U \defeq \sum_{n=0}^{\ttt_{\lambda}} U( \aaa_n; \s_n ).  
\end{align}
As we will later investigate $\lambda \rightarrow 1^{-}$ limit, we will ensure convergence by dividing the total utility by the average number of rounds played, $(1-\lambda)^{-1}$. The rescaled utilities are denoted $U_{\lambda}(a;s)$ and the rescaled total utility by $\U_{\lambda}$.

The goal of this paper is to classify all behavioural strategies of player $1$ which fix the expectation of $\U$ regardless of the behavioural strategy employed by player $2$. Below we will refer to player $1$ as the \textbf{autocrat} and payer $2$ as the \textbf{opponent}. The autocrat is so-called because in an asymmetric game, an advantaged player will be able to fix a large range of value than their opponent. To state our main result, we require the following definitions.
\begin{defn}
A behavioural strategy $\ppp$ fixes the expectation of $\U$ if there exists a constant $C$ such that $\E \U = C$ regardless of the behavioural strategy $\qqq$ employed by the opponent. 
\end{defn}

We will see that there is no need to consider histories that cannot be realized for a given strategy $\p$ of the autocrat, as these will not be seen by the expectation. This motivates the following definition.
\begin{defn}
\label{p_accessible}
The $\ppp$\textbf{-accessible n-histories}, denoted $\H_{\ppp,n}$, is the set of histories for which there exists a behavioural strategy $\qqq$ of player $2$ which ensures that said history occurs with positive probability . That is, 
\begin{align*}
\H_{\ppp,n} = \bigcup_{\qqq} \Big\lbrace h_{<n} \in \H_n \Big| \P_{\ppp,\qqq}(\h_{<n} = h_{<n}) > 0 \Big\rbrace  .
\end{align*}
\end{defn}

We will require amenable expressions for the conditional expectation of $\U$. The independence of $(\aaa_n)_{n \geq 0}$ with $\ttt_{\lambda}$ implies that
\begin{align*}
\frac{\E \U}{\E \ttt_{\lambda}} = \E \U_{\lambda} = \sum_{k=0}^{\infty} \E \left[ \sum_{n=0}^{k}U_{\lambda} (\aaa_n ; \s_n, \h_{<n}) \right] \P \left( \ttt_{\lambda} = k \right) 
 = (1-\lambda) \sum_{k=0}^{\infty} \lambda^k  \E \left[ \sum_{n=0}^{k}U_{\lambda} (\aaa_n; \s_n, \h_{<n}) \right] .
\end{align*}
As it is assumed $U_{\lambda}$ is bounded, switching the order of the summations 
\begin{align}
\label{discount_motivation}
(1-\lambda) \sum_{n=0}^{\infty}   \E \left[ U_{\lambda} (\aaa_n ; \s_n ,\h_{<n}) \right] \sum_{k=n}^{\infty} \lambda^k ,
\end{align}
yields
\begin{align}
\label{E_U}
\E \U_{\lambda}  = \sum_{n=0}^{\infty} \lambda^n   \E \left[ U_{\lambda} (\aaa_n ; \s_n, \h_{<n}) \right]  .
\end{align}
Note that the contribution of $U_{\lambda} (\aaa_n; \s_n,\h_{<n})$ to $\E \U$ is discounted by a factor of $\lambda^n$. We define $\U_{\geq n}$ to be $\U_{\lambda}$ with the first $n$ terms truncated,
\begin{align}
\U_{\geq n} \defeq \sum_{n \leq m \leq \ttt_{\lambda}} U_{\lambda} (\aaa_m; \s_m , \h_{<m}).
\end{align}
For each $n \geq 0$, we define the \textbf{future expected utility} function, $ \Phi_n : \left(\A \times \SSS \right) \times \H_n  \rightarrow \R $, to be the expectation of $\U_{\geq n}$ conditioned on round $n$ being reached with history $h_{<n+1}$. That is, 
\begin{align}
\label{phi_defn}
\Phi_n ( a_n ; s_n , h_{<n}) & \defeq \E  \Big[ \U_{\geq n} \Big| \ttt_{\lambda} \geq n , (\aaa_n; \s_n,\h_{<n}) = (a_{n};s_n,h_{<n})   \Big] 
\end{align}
for all $ (a_n;s_n,h_{<n})$ in $\left(\A \times \SSS \right) \times \H_{n}$. The same calculation as for $ \E \U$ yields
\begin{align}
\label{E_phi_n}
\Phi_n (a_{n};s_n,h_{<n}) =  \sum_{m=n}^{\infty} \lambda^{m-n} \E  \Big[  U_{\lambda} (\aaa_{m}; \s_m, \h_{<m})  \Big|  (\aaa_{n};\s_n,\h_{<n}) = (a_{n} ; s_n , h_{<n})  \Big]
\end{align}
for all $(a_{n}; s_n, h_{<n})$. We note that 
\begin{align*}
\E \U_{\lambda} = \E \Big[ \Phi_0 (\aaa_0 ; \s_0) \Big| \s_0 = s_0 \Big]. 
\end{align*}

With these definitions in hand, we can know state our first theorem.

\begin{theorem}[Expectation fixing strategies for discounted games]
\label{THM_Phi_condition}
Consider the two-player temporally discounted game, starting in state $s_0$, on the finite game graph $(\SSS,\A,U,\lambda)$. Player $1$'s behavioural strategy $\ppp$ fixes the expectation of $\lambda \U_{\lambda}$ iff the expression
\begin{align}
\label{Phi_condition}
\sum_{x \in \X_{s}} \Phi_{n} \left( x, y ;s , h_{< n} \right) p_{x| s,h_{< n}} = \sum_{x \in \X_{s}} \Phi_{n} \left( x, y' ;s , h_{< n} \right) p_{x| s,h_{< n}} 
\end{align}
for all $y,y' \in \Y_s$, $s \in \SSS$ and $h_{< n} \in \H_{\p,n}$. Moreover, in this case
\begin{align*}
\E \U_{\lambda} = \sum_{ x \in \X_{s_0} } \Phi_{0} \left( x,y; s_0 \right) p_{x|s_0} \quad \forall \: y \in \Y_{s_0}.
\end{align*}
\end{theorem}

\begin{rem}
In fact, a game with any number of players can be considered. In this setting, we call a set of players $\J \subset \I$ a \textbf{coalition}. The remaining players are collected in coalition $\K \defeq \I \setminus \J$. Let $\p^j$ and $\q^k$ represent the behavioural strategies for $j \in \J$ and $k \in \K$. We can view $\G \left( \A^{\J} \times \A^{\K} \right)$ as a two player game between coalitions with action spaces
\begin{align*}
\A^{\J} \defeq \prod_{j \in \J} \A^j \quad \mbox{and} \quad \A^{\K} \defeq \prod_{k \in \K} \A^k 
\end{align*}
and behavioural strategies $\p^{\J} \defeq \prod_{j \in \J} \p^j$ and $\q^{\K} \defeq \prod_{k \in \K} \q^k $.
\end{rem}

The proof is split up between Lemmas \ref{first_key_lemma} and \ref{second_key_lemma}. To tackle these lemmas, we need some machinery for the future expected utilities $\left(\Phi_n \right)_{n \geq 0}$. Particularly, their characterization as a solution to a recursive formula which satisfies a growth condition.

\begin{lem}[Bellman-type qquation]
\label{lem_phi_mn}
The expected future utilities $\left(\Phi_n : (\A \times \SSS ) \times \H_n  \rightarrow \R \right)_{n \geq 0}$, as defined in (\ref{phi_defn}), satisfy the one-step recursion formula
\begin{align}
\label{onestep_beta}
\Phi_n (a;s,h_{<n})= U_{\lambda} \left(a;s,h_{<n} \right) +\lambda \E \bigg[ \Phi_{n+1} (\aaa_{n+1};\s_{n+1},\h_{<n+1}) \bigg| \h_{<n+1} =(a;s,h_{<n})  \bigg]
\end{align}
for all $(a,s,h_{<n})$
\end{lem}
\begin{proof}
Recall the form of $\Phi_{n}$ from (\ref{E_phi_n}),
\begin{align*}
\Phi_n (a;s,h_{<n}) &= U_{\lambda} (a;s,h_{<n}) +   \sum_{k=n+1}^{\infty} \lambda^{k-n} \E \Big[ U_{\lambda}(\aaa_k;\s_k,\h_{<k}) \Big| (\aaa_n,\s_n,\h_{<n}) = (a;s,h_{<n}) \Big] .
\end{align*}
By conditioning on $(\aaa_{n+1},\s_{n+1})$, we see that $\E \Big[ U_{\lambda}(\h_{\leq k}) \Big| \h_{\leq n} = (a;s,h_{<n}) \Big]$ equals
\begin{align*}
\sum_{(a',s')} \E \Big[ U_{\lambda}(\aaa_{k};\s_k,\h_{<k}) \Big| \h_{\leq n+1}= (a';s',h_{< n+1}) \Big] P_{n+1} \big( a'   \big| a,s,h_{<n} \big).
\end{align*}
Switching the order of summations and noting that
\begin{align*}
\Phi_{n+1}(a';s',h_{<n+1}) = \sum_{k=n+1}^{\infty} \lambda^{k-(n+1)}  \E \Big[  U_{\lambda}(\aaa_{k};\s_k,\h_{<k}) \Big| \h_{\leq n+1}= (a';s',h_{< n+1}) \Big]
\end{align*}
lets us rewrite the double sum as
\begin{align*}
\lambda \sum_{(a',s')}   \Phi_{n+1}(a';s',h_{<n+1}) P_{n+1} \big( a'   \big| a,s,h_{<n} \big)
\end{align*}
Thus we have
\begin{align*}
\Phi_n (a;s,h_{<n}) &= U_{\lambda} (a;s,h_{<n}) +  \lambda \E \Big[ \Phi_{n+1} (\aaa_{n+1};\s_{n+1},\h_{<n+1}) \Big| (\aaa_n ;\s_n ,\h_{<n}) =(a;s,h_{<n}) \Big].
\end{align*} 
\end{proof}
\noindent It will be useful in the proofs of Lemmas \ref{first_key_lemma} and \ref{second_key_lemma} to rewrite the expectation in equation (\ref{onestep_beta}) in terms of the behavioural strategies $\p$ and $\q$. That is,
\begin{align*}
\sum_{x \in \X_{s}} \sum_{y \in \Y_{s}} \Phi_{n+1} \left( x,y ;s , h_{< n} \right) p_{x| s,h_{< n}}q_{y|s,h_{<n}} = \lambda^{-1}  \left[ \Phi_{n} ( x_{n},y_n; s_{n} , h_{< n}) - U_{\lambda}( x_{n},y_{n};s_{n}) \right] 
\end{align*}
for all $(x_n,y_n;s_n)$ in $\A_{s_n} \times \SSS$. 

By definition, $\Phi_n$ does not depend on $\q_m$ for $m \leq n$. If, in addition, the autocrat's strategy $\p$ ensures that the expected future utility $\E \Phi_n (\x_n,y,h_{<n})$ is independent of the opponent's action $y$ for every $\p$-accessible history $h_{<n}$, we can conclude that the $\Phi_n$'s do not depend of $\q$. This is proved in the following lemma.  

\begin{lem}[Sufficient direction]
\label{first_key_lemma}
Suppose that the expectation of $\Phi_n (\x_n , y;\s_n, \h_{<n})$,
\begin{align*}
\sum_{x \in \X_s} \Phi_{n} (x,y;s,h_{<n}) p_{x|s,h_{<n}} ,
\end{align*}
is constant in $y$ for all $h_{<n} \in \H_{\ppp,n}$ and $n \geq 0$. Then $\Phi_n$ does not depend on $\qqq$ for all $n$.
\end{lem}
\begin{proof}
Recall the one-step recursive formula,
\begin{align*}
\Phi_{n-1} (h_{<n}) =   U_{\lambda}(h_{<n}) + \lambda \sum_{y_n } \left( \sum_{x_n} \Phi_{n} (x_n,y_n; s_n,h_{<n})   p_{x_n|s,h_{<n}} \right) q_{y_n|s,h_{<n}} .
\end{align*}
We have
\begin{align*}
\Phi_{n-1} (h_{<n}) =   U_{\lambda}(h_{<n}) + \lambda \sum_{x_n }\Phi_{n} (x_n,y;s_n,h_{<n})   p_{x_n|s,h_{<n}}   \quad \forall \: y \in \Y_s.
\end{align*}
That is, $\Phi_n$ does not depend on $\qqq_n$. Repeating this, we see
\begin{align*}
\Phi_{n} (h_{<n+1}) =   U_{\lambda}(h_{<n+1}) + \lambda  \sum_{x_{n+1}} \Phi_{n+1} (x_{n+1},y; s_{n+1},h_{<n+1})   p_{x_{n+1}|s_{n+1},h_{<n+1}}  .
\end{align*}
That is, $\Phi_n$ does not depend on $\qqq_{n+1}$. We can do this for all $\q_m$ for $m \geq n$.
\end{proof}
\begin{lem}[Necessary direction] 
\label{second_key_lemma}
Suppose player $1$'s strategy $\p$ fixes the expectation of $\U$. Then 
\begin{align*}
\sum_{x \in \X_s} \Phi_{n} (x,y;s,h_{<n}) p_{x|s,h_{<n}} ,
\end{align*}
is constant in $y$ for all $h_{<n} \in \H_{\ppp,n}$ and $n \geq 0$
\end{lem}

\begin{proof}
For any $n$, we can write $\E \U_{\lambda}$ as
\begin{align*}
\E \U_{\lambda} = \sum_{m=0}^{n-1} \E \left[ U_{\lambda} ( \h_{\leq m} ) \right] + \lambda^n \E \left[ \Phi_n (\aaa_n ;\s_n,\h_{<n}) \right].  
\end{align*}
This can be rewritten explicitly in terms of the behavioural strategies $\p$ and $\q$. That is,
\begin{align*}
\E \U_{\lambda} = \sum_{m=0}^{n-1} \lambda^m \sum_{h_{\leq m} } U_{m} (h_{\leq m}) \prod_{k=0}^{m} p_{x_k|s_k,h_{<k}} q_{y_k |s_k, h_{<k}} 
+ \lambda^{n} \sum_{h_{\leq n} } \Phi_{n} (h_{\leq n}) \prod_{k=0}^{n} p_{x_k | s_k, h_{<k}} q_{ y_k | s_k, h_{<k}}.
\end{align*}
Furthermore, the second term equals
\begin{align*}
 \sum_{ h_{<n} \in \H_{\p,n}} \left[ \sum_{y_n} \left( \sum_{x_n} \Phi_{n} (x_n,y_n;s_n,h_{< n}) p_{x_n |s_n, h_{<n}} \right) q_{y_n |s_n, h_{<n}} \right] \P (\aaa_{<n} = h_{<n}) .
\end{align*}
By assumption, there is some constant $C$ such that $\E \U_{\lambda} = C$ regardless of the $\q$. Differentiating with respect to $q_{y|s_n,h_{<n}}$ and dividing out $\P (\aaa_{<n} = h_{<n})$ yields
\begin{align*}
\sum_{x_n \in \X_{s_n}}  \big[ \Phi_{n} (x_n,y;s_n,h_{< n}) - \Phi_{n} (x_n,y';s_n,h_{< n}) \big] p_{x_n |s_n, h_{<n}}  = 0 \quad \forall \: y' \in \Y_{s_n} .
\end{align*}
\end{proof}
\noindent As demonstrated in the proof above, when condition (\ref{Phi_condition}) holds, equation (\ref{onestep_beta}) simplifies to 
\begin{align}
\label{RECURSION}
\sum_{x \in \X_{s}} \Phi_{n+1} \left( x,y ;s , h_{< n} \right) p_{x| s,h_{< n}} = \lambda^{-1}  \left[ \Phi_{n} ( x_{n},y_n; s_{n} , h_{< n}) - U_{\lambda}( x_{n},y_{n};s_{n}) \right] .
\end{align}
This will be called the \textbf{one-step recursion formula}.

\section{Enforceable Values of $\E \U$ and Strategies}

\subsection{Autocratic Strategies}

The goal of this section is to provide both a characterization of the enforceable values $[m_s,M_s]$ of the expected total utility and the autocratic strategies which enforce them. We will first restrict to case where the expected future utility does not depend on the the opposing player's current action. That is,
\begin{align}
\label{CONDITION_opposition_agnostic}
\Phi_k (x,y;s,h_{<k}) \equiv \Phi_k (x;s,h_{<k}) .
\end{align}
The values that can be enforced, despite this restriction, will be referred to as \textbf{opposition agnostic}. The same will be said of their respective autocratic strategies. This assumption is crucial to our main result, that opposition agnostic autocratic strategies need only finite memory.

This claim can be refined further. It will turn out that the autocrat can fix any opposition agnostic value by playing mixed strategies with state-dependent two-action support. A state is called \textbf{cornered} if the autocrat must play one action to fix a value. In general, cornered states form directed trees. We note it is possible to construct games in which there are long paths of cornered states, as in Example \ref{EXMP_agnostic_pathological}. Our main result is then stated as follows.
\begin{theorem}[Autocratic strategies with finite memory]
\label{THM_main_result}
Suppose the multi-state game $\G(\SSS,\A,U)$ is finite. Then any opposition agnostic value of the expected total utility can be enforced by a finite memory autocratic strategy. In particular, it requires $L+1$-memory, where $L$ is the length of the longest path of cornered states.
\end{theorem}  

The proof of Theorem \ref{THM_main_result} relies on a characterization of the extremal enforceable values as the solutions to a system of fixed point equations indexed by the states of a maximal subgraph. In order to state this theorem we require the following definition.
\begin{defn}
\label{prune}
A subgraph $(\SSS',\A')$ of the game graph $(\SSS,\A)$ has been \textbf{pruned} if edge $(x,y)$'s absence implies $(x,y')$ is absent for all $y'$ and all states have at least one outgoing edge. 
\end{defn}
More is said about pruning in the lead-up to Algorithm \ref{ALGO}. For now it is enough to associate a pruned subgraph with the game $\G(\SSS,\A,U)$ in which the autocrat is constrained to play mixed strategies with nonempty support contained in some $\X'_s$ when in state $s$. With this definition in hand we can state our characterization of the enforceable values.

\begin{theorem}[Characterization of enforceable values]
\label{THM_characterization_opposition_agnostic}
The opposition agnostic enforceable values of the game $\G(\SSS,\A,U)$ are the intervals $[m_s, M_s]$ given by the unique solution to
\begin{align}
\label{BIG}
\begin{split}
m_s &= \min_{x \in \X^f_{s}} \max_{y \in \Y_{s}} \big\{ \lambda m_{T(x,y;s)} + U_{\lambda} (x,y;s)  \big\} \quad \forall \: s \in \SSS' \\
M_s &= \max_{x \in \X^f_{s}} \max_{y \in \Y_s} \big\{ \lambda M_{T(x,y;s)} + U_{\lambda} (x,y;s)  \big\} \quad \forall \: s \in \SSS'
\end{split}
\end{align}
of the maximal pruned subgraph $\G (S', \X^f \times \Y)$ such that
\begin{align}
\label{main_inequality}
\max_{y \in \Y_s} \big\{ \lambda m_{T(\x_s,y;s)} + U_{\lambda} (x,y;s)  \big\} \leq  \min_{y \in \Y_s} \big\{ \lambda M_{T \left( \x_s,y;s \right)} + U_{\lambda} (x,y;s) \big\},
\end{align}
for any extremal action $\x_s$ as defined in (\ref{x_extremal_left}) and (\ref{x_extremal_right}).
\end{theorem}
\noindent The use of the term maximal is justified by the following lemma. The proof of this and the main result is deferred to Section \ref{proof_main_results}.
\begin{lem}
Suppose $\left( [m_s ,M_s] \right)_{\SSS_1}$ and $\left( [m'_s ,M'_s] \right)_{ \SSS_2}$ are the solutions to the system of equations (\ref{BIG}) and (\ref{main_inequality}) for pruned game subgraphs $\G_1$ and $\G_2$. Then the solution $\left( [m_s^*,M_s^*]\right)_{ \SSS_1 \cup \SSS_2}$  to (\ref{BIG}) and (\ref{main_inequality}) on $\G_1 \cup \G_2$ satisfies $m_s^* \leq \min \{ m_s, m_s' \}$ and $M_s^* \geq \min \{ M_s, M_s' \}$ for all $s$ in $\SSS_1 \cup \SSS_2$ . 
\end{lem}

The opposition agnostic autocratic strategies can be written down in terms of the extremal enforceable values of Theorem \ref{THM_characterization_opposition_agnostic}. These strategies require only the left and right \textbf{extremal actions}, $x^-_s$ and $x^+_s$ in each state. They are defined for some $\X' \subset \X$ as
\begin{align}
\label{x_extremal_left}
x^-_s & \defeq \argmin_{x \in \X'_s} \max_{y \in \Y_s} \big\{ \lambda m_{T (x,y;s)} + U_{\lambda} (x,y;s) \big\}
\end{align}
and
\begin{align}
\label{x_extremal_right}
x^+_s & \defeq \argmax_{x \in \X'_s} \min_{y \in \Y_s} \big\{ \lambda M_{T (x,y;s)} + U_{\lambda} (x,y;s) \big\} .
\end{align}
When there are no cornered states (that is, $x^+_s \neq x^-_s$ for each $s$), memory-$1$ strategies can be constructed. Recall the condensed notation (\ref{cond_prob_succinct}) for the conditional probability. The first round strategy is the solution to
\begin{align*}
M_s P_0 \big[ x^+_{s_0} \big|  s_0 \big] + m_s P_0 \big[ x^-_{s_0} \big|  s_0 \big] = m .
\end{align*}
In subsequent round $i+1$, the memory-$1$ strategy when in state $s$ can be found as follows. Suppose that in round $i$ the state was parent $s_p$ of $s$, the joint action was $(x^+_p,y_p)$, and $\Phi_i (x^+_p,y_p;s_p,h_{<i}) = M_{s_p}$. Since $\X'=\X_f$, the inequality (\ref{main_inequality}) implies that 
\begin{align*}
m_s \leq \lambda^{-1} \left[ M_{s_p} - U_{\lambda} \left(x^+_p,y_p;s_p \right) \right] \leq M_s.
\end{align*}
Thus, given the history $(x^+_p,y_p;s_p,h_{<i})$ in state $s$, the behavioural strategy in round $i+1$ can be found by solving
\begin{align*}
M_s P_{i+1} \big[ x^+_{s} \big| (s, x^+_p,y_p,s_p)  \big] + m_s P_{i+1} \big[  x^-_{s} \big|  (s, x^+_p,y,s_p)  \big] = \lambda^{-1} \left[ M_{s_p} - U_{\lambda} \left(x^+_p,y_p;s_p \right) \right] .
\end{align*}
This is illustrated in Figure \ref{DGRAM_mem1_strategy}. A similar formula holds for parent states and joint action pairs that satisfy $s =T \big( x^-_p,y_p;s_p \big)$. This proves the following.

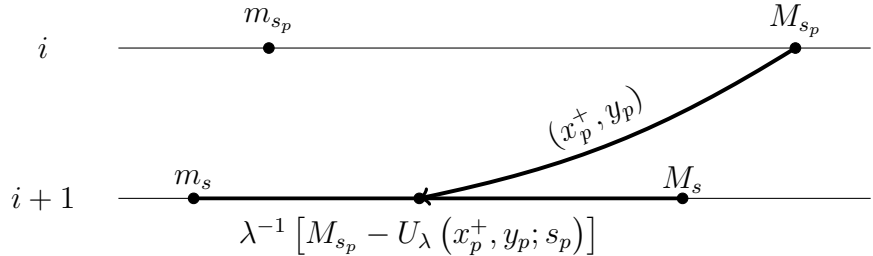
\begin{figure}[h!]
   \centering   
   \begin{tikzpicture}[decoration=brace]


\draw(0,0)--(10,0);

\filldraw [black] (2,0) circle (2pt);
\filldraw [black] (9,0) circle (2pt);

\node [fill=none] at (2,2ex) (2) {$m_{s_p}$};

\node [fill=none] at (9,2ex) (3) {$M_{s_p}$};

\node [fill=none] at (-1,0) (4) {$i$};


\node [fill=none] at (-1,-2) (8) {$i+1$};

\draw(0,-2)--(10,-2);
\draw[line width=0.5mm] (1,-2) --  (7.5,-2) ;

\filldraw [black] (4,-2) circle (2pt);
\filldraw [black] (1,-2) circle (2pt);
\filldraw [black] (7.5,-2) circle (2pt);

\node [fill=none] at (4,-2.5) (5) {$\lambda^{-1} \left[ M_{s_p} - U_{\lambda} \left(x^+_p,y_p;s_p \right) \right]$};

\node [fill=none] at (1,-1.75) (6) {$m_{s}$};

\node [fill=none] at (7.5,-1.75) (7) {$M_{s}$};


\draw[->,line width=1.5pt,bend left=10] (9,0) to node[pos=0.35,sloped, above left] {$(x_p^+,y_p)$} (4,-2);

 \end{tikzpicture}
   \caption{one step}
   \label{DGRAM_mem1_strategy}
\end{figure}

\begin{theorem}[Memory-$1$ autocratic strategies]
\label{THM_mem1_strategy}
Suppose the maximal admissible subgraph $\G(\SSS,\X'\times \Y)$ of Theorem \ref{THM_characterization_opposition_agnostic} satisfies $ \left| \X'_s \right| \geq 2$ for all $s \in \SSS'$. Then, for a game with initial state $s_0$, there exists a memory-$1$ autocratic strategy for each element of $m \in [m_{s_0},M_{s_0}]$. Furthermore, the strategy need only play actions $x^+_s$ and $x^-_s$ in each state. 
\end{theorem}
\noindent Unfortunately, when there are chains of cornered states, autocratic strategies require memory greater than the length of the longest chain.   
\begin{proof}[Proof of Theorem \ref{THM_main_result}]
The proof proceeds by construction. If state $s$ is cornered, the autocrat must play a single action. If state $s$ has uncornered parents, we use the formulas developed in the proof of Theorem \ref{THM_mem1_strategy}. Finally, it remains to deal with uncornered states which have a cornered parent.

Suppose that in round $n+l+1$ the autocrat finds themselves in the uncornered state $s$ with a history 
\begin{align}
\label{history_THM_main_result}
h_{<n+l+1} = \left( (a_l;s_l) ,(a_{l-1};s_{l-1}), \ldots , (a_1;s_1), (x_r^+,y_r;r), h_{<n} \right) 
\end{align}
of cornered states $s_1$ through $s_l$ and the uncornered state $r$. That is, the actions $x_1$ through $x_l$ are extremal actions. By assumption $l \leq L$. Furthermore, suppose 
\begin{align*}
\Phi_{n} \left( x^+_r,y_r;s_r,h_{<i} \right) = M_{r}.
\end{align*}
Since $\X'=\X_f$, the inequality (\ref{main_inequality}) implies that
\begin{align*}
\Phi_{n+i} \left( a_{i};s_{i},  \ldots ,a_1,s_1, x_r^+,y_r,s_r,h_{<i} \right) = \lambda^{-i}M_r - \lambda^{-i} U(x_r^+,y_r;r) - \sum_{k=1}^{i-1} \lambda^{-i+k} U_{\lambda}(a_k;s_k) 
\end{align*}
lies in $[ m_{s_i}, M_{s_i}]$ for $1 \leq i \leq l$ and that
\begin{align*}
 m_s \leq \lambda^{-1} \left[ \Phi_{n+l} \left( a_{l};s_{l},  \ldots ,a_1,s_1, x_r^+,y_r,s_r,h_{<i} \right) - U_{\lambda}(a_l;s_l) \right] \leq M_s .
\end{align*}
Thus, given the history (\ref{history_THM_main_result}) in state $s$, we can solve for the memory-$(L+1)$ behavioural strategy in round $n+l+1$ by equating
\begin{align*}
M_{s} P_{n+l+1} \big[ x_s^+ \big| a_l,s_l, \ldots a_1,s_1,x^+_r,y_r, r \big] + m_s P_{n+l+1} \big[ x_s^- \big| a_l,s_l, \ldots a_1,s_1,x^+_r,y_r, r  \big] 
\end{align*}
with the above.
\end{proof}

\subsection{The Characterization of Enforceable Values}
\label{proof_main_results}

For a fixed game graph $\G(\SSS, \A,U)$, the left (and right) opponent agnostic enforceable values can be characterized as the solutions to a system of fixed point equations (FPEs).

\begin{theorem}[Enforceable values] 
\label{THM_enforceable}
Suppose the autocrat can only play strategies with support $\X' \subset \X$. If the enforceable values $[m_s,M_s]$ of $\E \U_{\lambda}$ exist on the game graph $~{(\SSS',\X' \times \Y,U)}$, they are the unique solution to the fixed point equations
\begin{align}
\label{EQN_left_FPE}
m_{s} \defeq \min_{x \in \X'_s} \max_{y \in \Y_s} \big\{ \lambda m_{T(x,y;s)} + U_{\lambda} (x,y;s)  \big\}  \quad \forall \: s \in \SSS'.
\end{align}
and
\begin{align}
\label{EQN_right_FPE}
M_{s} \defeq \max_{x \in \X'_s} \min_{y \in \Y_s} \big\{ \lambda M_{T(x,y;s)} + U_{\lambda} (x,y;s)  \big\} \quad \forall \: s \in \SSS'.
\end{align}
and the inequalities
\begin{align}
\max_{y \in \Y_s} \big\{ \lambda m_{T(\x_s,y;s)} + U_{\lambda} (x,y;s)  \big\} \leq  \min_{y \in \Y_s} \big\{ \lambda M_{T \left( \x_s,y;s \right)} + U_{\lambda} (x,y;s) \big\},
\end{align}
for all extremal actions $\x_s$ (see (\ref{x_extremal_left}) and (\ref{x_extremal_right})) .
\end{theorem}
\noindent A first attempt the proof involves an iterative but unstable approximation of the the endpoints. Defining
\begin{align}
\label{U_bounds}
m_{0} \defeq \min_{a,s} \big\{ U (a;s) \big\} \quad \mbox{and} \quad M_{0} \defeq \max_{a,s} \big\{ U (a;s) \big\},
\end{align}
we take the initial exterior estimate of the enforceable values in each state to be $\left[m_{s,0}, M_{s,0} \right] \defeq \left[ m_0,M_0 \right]$. A proposed value of the expectation $v$ must satisfy the first step of the recursive Bellman equation (\ref{eqn_outline_bellman}). That is,
\begin{align*}
v = \sum_{x \in \X_{s}} \Phi_{0} \left( x,y ;s  \right) p_{x| s} .
\end{align*}
The potential value $v$ is discarded if, for any $\Phi_0$ and $p_{\cdot|s}$ satisfying the above, there exists a joint action $(x,y)$ in $s$ such that
\begin{align*}
\lambda^{-1} \left[\Phi_{0} \left( x,y ;s  \right) - U_{\lambda} (x,y;s) \right] \notin \big[m_{T(x,y;s),0},M_{T(x,y;s),0} \big].
\end{align*}
Using the opponent agnostic assumption, $\Phi_0(x,y;s) \equiv \Phi_0 (x;s)$ , the above is equivalent to
\begin{align*}
\Phi_0 (x;s) \notin \lambda \big[m_{T(x,y;s),0},M_{T(x,y;s),0} \big] + U_{\lambda} (x,y;s) \quad \forall \; y \in \Y_s ,
\end{align*}
which can be rewritten as
\begin{align*}
\max_{y \in \Y_s} \left\{ \lambda m_{T(x,y;s),0} + U_{\lambda} (x,y;s) \right\} \leq \Phi_0 (x;s) \leq \min_{y \in \Y_s} \left\{ \lambda M_{T(x,y;s),0} + U_{\lambda} (x,y;s) \right\}.
\end{align*}
Letting $ \X_{s,1} \subset \X_s $ be all $x$ such that the above inequality holds, any enforceable value must then lie in
\begin{align*}
\left[ m_{s,1},M_{s,1} \right] \defeq \left[ \min_{x \in \X_{s,1}} \max_{y \in \Y_s} \left\{ \lambda m_{T(x,y;s),0} + U_{\lambda} (x,y;s) \right\}, \max_{x \in X_{s,1 }}  \min_{y \in \Y_s} \left\{ \lambda M_{T(x,y;s),0} + U_{\lambda} (x,y;s) \right\} \right]
\end{align*}
This procedure is iterated, giving us the recursive formulas for the autocrat's action sets
\begin{align*}
\X_{s,n+1} \defeq \bigg\{ x \in \X_{s,n} \bigg| \max_{y \in \Y_s} \left\{ \lambda m_{T(x,y;s),n} + U_{\lambda} (x,y;s) \right\}  \leq \min_{y \in \Y_s} \left\{ \lambda M_{T(x,y;s),n} + U_{\lambda} (x,y;s) \right\} \bigg\}
\end{align*}
and the extremal values
\begin{align}
\label{m_recursion_unstable}
m_{s,n+1} \defeq \min_{x \in \X_{s,n}} \max_{y \in \Y_s} \big\{ \lambda m_{T(x,y;s),n} + U_{\lambda} (x,y,s)  \big\}
\end{align}
and 
\begin{align}
\label{M_recursion_unstable}
M_{s,n+1} \defeq \max_{x \in \X_{s,n}} \min_{y \in \Y_s} \big\{ \lambda M_{T(x,y;s),n} + U_{\lambda} (x,y,s)  \big\}.
\end{align}
Note that the above sequences are monotonic, when they exist.

Note $\X_{s,n+1} \subset \X_{s,n}$. Taking $n \rightarrow \infty$ we can define
\begin{align*}
m_s \defeq \lim_{n \rightarrow \infty} m_{s,n} \quad \mbox{and} \quad M_s \defeq \lim_{n \rightarrow \infty} M_{s,n} \quad \mbox{and} \quad \X^f_s \defeq \lim_{n \rightarrow \infty} \X_{s,n} = \cap_{n \geq 0} \X_{s,n}.
\end{align*}
For each state $s$ we get either an empty set $\X_s^f$ and no enforceable values or a nonempty set autocrat actions
\begin{align*}
\X^f_s = \bigg\{ x \in \X_s \bigg| 
\max_{y \in \Y_s} \left\{ \lambda m_{T(x,y;s)} + U_{\lambda} (x,y;s) \right\}  \leq \min_{y \in \Y_s} \left\{ \lambda M_{T(x,y;s)} + U_{\lambda} (x,y;s) \right\}  \bigg\}
\end{align*}
and a nonempty interval
\begin{align*}
[m_{s},M_s] = \left[ \min_{x \in \X^f_{s}} \max_{y \in \Y_s} \big\{ \lambda m_{T(x,y;s)} + U_{\lambda} (x,y,s)  \big\},
 \max_{x \in \X^f_{s}} \min_{y \in \Y_s} \big\{ \lambda M_{T(x,y;s)} + U_{\lambda} (x,y,s)  \big\} \right]
\end{align*}
of enforceable values.

%


As a stable algorithm is ideal, we must contend with the fact that the above process is irreversible. In particular, once an action $x$ is excluded from some $\X_{s,N}$, it is excluded from $\X_{s,n}$ for all $n \geq N$. Instead we formulate an algorithm with a globally convergent step, detailed below.

We will first recursively define a sequence of game subgraphs $(\SSS_i, \X^i \times \Y,U)$ such that $\SSS_{i+1} \subset \SSS_i$ and $\X^{i+1} \subset \X^i$ in the following manner. Let $\SSS_0 = \SSS$ and $\X^0 = \X$. For any particular game subgraph $(\SSS_i, \X^i \times \Y,U)$ we can recursively define monotonic sequences
\begin{align}
\label{m_recursion}
m_{s,n+1} \defeq \min_{x \in \X^i_s} \max_{y \in \Y_s} \big\{ \lambda m_{T(x,y;s),n} + U_{\lambda} (x,y;s)  \big\}  \quad \forall \: s \in \SSS_i.
\end{align}
and
\begin{align}
\label{M_recursion}
M_{s,n+1} \defeq \max_{x \in \X^i_s} \min_{y \in \Y_s} \big\{ \lambda M_{T(x,y;s),n} + U_{\lambda} (x,y;s)  \big\} \quad \forall \: s \in \SSS_i 
\end{align}
with $M_{s,0}=M_0$ and $m_{s,0}= m_0$. We will show in Section \ref{SECTION_algo} that the above are globally convergent. Furthermore, these sequences are defined for all $n$, as opposed to those defined above in (\ref{m_recursion_unstable}) and (\ref{M_recursion_unstable}). Monotonicity implies the limits $m_s(\X^i)$ and $M_s(\X^i)$ exist. But note the inequality
\begin{align}
\label{figure_inequality}
\max_{y \in \Y_s} \big\{ \lambda m_{T(\x_s,y;s)} + U_{\lambda} (\x_s,y;s)  \big\} \leq \min_{y \in \Y_s} \big\{ \lambda M_{T(\x_s,y;s)} + U_{\lambda} (\x_s,y;s)  \big\} 
\end{align} 
may not hold for some extremal action $\x_s$. If this is the case, then regardless of the value that $\Phi_n(\x_s;s)$ takes there exists an opponent action $y$ such that \begin{align*}
\lambda^{-1} \left[ \Phi_n(\x_s;s) - U_{\lambda}(\x_s,y;s) \right] \notin \left[ m_{T(\x_s,y;s)}, M_{T(\x_s,y;s)} \right] .
\end{align*} 
An example of this is illustrated below in Figure \ref{DGRAM_neccessity_of_inequality}.

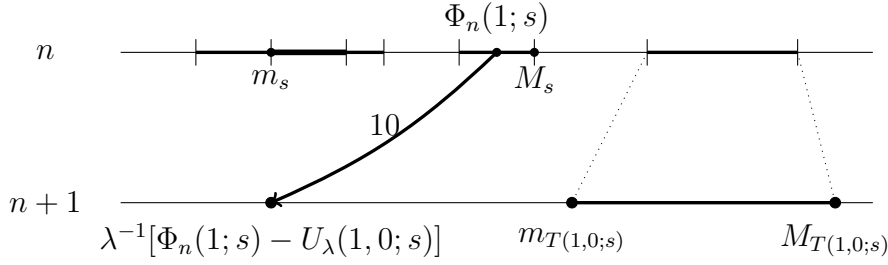
\begin{figure}[h!]
   \centering   
   \begin{tikzpicture}[decoration=brace]


\draw(0,0)--(10,0);

\draw[line width=0.45mm](1,0)--(3,0);
\draw[line width=0.45mm](2,0)--(3.5,0);
\draw[line width=0.75mm](2,0)--(3,0);

\draw[line width=0.45mm](4.5,0)--(5.5,0);
\draw[line width=0.45mm](7,0)--(9,0);

\filldraw [black] (2,0) circle (1.5pt);
\filldraw [black] (5.5,0) circle (1.5pt);

\filldraw [black] (5,0) circle (1.5pt);

\node [fill=none] at (2,-2ex) (2) {$m_s$};
\node [fill=none] at (5.5,-2ex) (2) {$M_s$};
\node [fill=none] at (5,2.5ex) (2) {$\Phi_n(1;s)$};

\node [fill=none] at (-1,0) (4) {$n$};

\draw(1,5pt)--(1,-5pt) node[below] {};
\draw(2,5pt)--(2,-5pt) node[below] {};
\draw(3,5pt)--(3,-5pt) node[below] {};
\draw(3.5,5pt)--(3.5,-5pt) node[below] {};

\draw(4.5,5pt)--(4.5,-5pt) node[below] {};
\draw(5.5,5pt)--(5.5,-5pt) node[below] {};
\draw(7,5pt)--(7,-5pt) node[below] {};
\draw(9,5pt)--(9,-5pt) node[below] {};


\draw(0,-2)--(10,-2);
\draw[line width=0.45mm] (6,-2) --  (9.5,-2) ;

\node [fill=none] at (-1,-2) (4) {$n+1$};

\filldraw [black] (6,-2) circle (2pt);
\filldraw [black] (2,-2) circle (2pt);
\filldraw [black] (9.5,-2) circle (2pt);

\node [fill=none] at (6,-2.5) (6) {$m_{T(1,0;s)}$};

\node [fill=none] at (9.5,-2.5) (7) {$M_{T(1,0;s)}$};

\node [fill=none] at (2,-2.5) (8) {$\lambda^{-1} [ \Phi_n (1;s) - U_{\lambda}(1,0;s)]$};

%
%
%
%
%
%
%


\draw[dotted] (7,0) to (6,-2);
\draw[dotted] (9,0) to (9.5,-2);

\draw[->,line width=1.25pt,bend left=10] (5,0) to node[pos=0.4, left] {$10$}(2,-2);

 \end{tikzpicture}
   \caption{Inequality (\ref{figure_inequality}) violated for extremal action $x_s^+=1$ and opponent action $y=0$.}
   \label{DGRAM_neccessity_of_inequality}
\end{figure}

If the inequality does not hold for some extremal actions $\x_s$, we remove the edges corresponding to said actions from $(\SSS_i,\X^i \times \Y)$. pruning (see Definition \ref{prune}) the resulting subgraph gives us the next iterations of both $\SSS_{i+1}$ and $\X^{i+1}$. Note that although the pruning steps are irreversible, we now have a globally convergent step in between successive prunings.

Finally, notice that $\X^{i+1}_s \subset \X^i_s$ implies
\begin{align*}
\max_{x \in \X^i_s} \min_{y \in \Y_s} \big\{ \lambda M_{T(x,y;s)} + U_{\lambda} (x,y;s)  \big\} \geq \max_{x \in \X^{i+1}_s} \min_{y \in \Y_s} \big\{ \lambda M_{T(x,y;s)} + U_{\lambda} (x,y;s)  \big\}
\end{align*}
and 
\begin{align*}
\min_{x \in \X^i_s} \max_{y \in \Y_s} \big\{ \lambda m_{T(x,y;s)} + U_{\lambda} (x,y;s)  \big\} \leq \min_{x \in \X^{i+1}_s} \max_{y \in \Y_s} \big\{ \lambda m_{T(x,y;s)} + U_{\lambda} (x,y;s)  \big\} .
\end{align*}
That is, the sequences $\left\{ (m_s(\X^i) \right\}_{i \geq 0}$ and $\left\{ (M_s(\X^i) \right\}_{i \geq 0}$ are monotonic.

\subsection{Uniqueness of the Fixed Point Equations}
\label{SECTION_uniqueness}


In order to prove the global convergence of the iterative processes (\ref{m_recursion}) and (\ref{M_recursion}) for a fixed game subgraph $(\SSS_i,\X^i \times \Y, U)$, it will be necessary to study their evolution when initialized with $m_{s,0}=M_0$ and $M_{s,0}=m_0$. We will denote these recursively defined sequences by $(\tilde{m}_{s,n})_{n \geq 0}$ and $(\tilde{M}_{s,n})_{n \geq 0}$. This has a natural interpretation related to enforceability. Suppose the autocrat did not want to enforce a value but rather ensure that the expected utility did not equal $v$ regardless of the opponent's strategy. These are the \textbf{excludable} values. 

The initial exterior estimate of the excludable values in each state state is taken to be $\left[\tilde{M}_{s,0}, \tilde{m}_{s,0} \right] \defeq \left[ m_0,M_0 \right]$. An excludable value $v$ must satisfy
\begin{align*}
v = \sum_{x \in \X_{s}} \Phi_{0} \left( x,y ;s  \right) p_{x| s} .
\end{align*}
The potential value $v$ is included if there exists $\Phi_0$ and $p_{\cdot|s}$ satisfying the above such that
\begin{align*}
\lambda^{-1} \left[\Phi_{0} \left( x,y ;s  \right) - U_{\lambda} (x,y;s) \right] \notin \big[\tilde{M}_{T(x,y;s),0},\tilde{m}_{T(x,y;s),0} \big]
\end{align*}
for all joint actions $(x,y)$. Using the opponent agnostic assumption, $\Phi_0(x,y;s) \equiv \Phi_0 (x;s)$ , the above is equivalent to
\begin{align*}
\Phi_0 (x;s) \notin \lambda \big[\tilde{M}_{T(x,y;s),0},\tilde{m}_{T(x,y;s),0} \big] + U_{\lambda} (x,y;s) \quad \forall \; y \in \Y_s ,
\end{align*}
which can be rewritten as
\begin{align*}
\Phi_0 (x;s) \leq \min_{y \in \Y_s} \left\{ \lambda \tilde{M}_{T(x,y;s),0} + U_{\lambda} (x,y;s) \right\} \: \mbox{or} \: \max_{y \in \Y_s} \left\{ \lambda \tilde{m}_{T(x,y;s),0} + U_{\lambda} (x,y;s) \right\} \leq \Phi_0 (x;s).
\end{align*}
Setting $p_{x|s}=1$, any excludable value must then lie outside
\begin{align*}
\left( \tilde{M}_{s,1}, \tilde{m}_{s,1} \right) \defeq \left( \max_{x \in \X_s} \min_{y \in \Y_s} \left\{ \lambda \tilde{M}_{T(x,y;s),0} + U_{\lambda} (x,y;s) \right\} , \min_{x \in \X_s} \max_{y \in \Y_s} \left\{ \lambda \tilde{m}_{T(x,y;s),0} + U_{\lambda} (x,y;s) \right\} \right).
\end{align*}
Iterating this procedure returns the recursions (\ref{m_recursion}) and (\ref{M_recursion}) and thus the fixed point equations (\ref{EQN_left_FPE}) and (\ref{EQN_right_FPE}).


Notice the opponent agnostic assumption leads to uncoupled recursive formulas for $m_s$ and $M_s$. We shall concern ourselves with the $m_s$ recursion and the lemmas required for proving the existence and uniqueness of the solution to (\ref{EQN_left_FPE}) in Theorem \ref{THM_enforceable}. The proof for (\ref{EQN_right_FPE}) is essentially identical, modulo the direction of the inequalities. 

To proceed with the proof, we require a few lemmas. First, we have the following crude bounds for the future expected utility.
\begin{lem}[Bounds for $\Phi$] 
The bounds on the utility in (\ref{U_bounds}) imply that
\begin{align}
\label{BOUNDS}
 m_0 \leq \Phi_n \leq M_0 \quad \forall \: n \geq 0 .
 \end{align}
\end{lem}
\begin{proof}
Equation (\ref{E_phi_n}) states 
\begin{align*}
\Phi_n (a_{n};s_n,h_{<n}) =  (1-\lambda)\sum_{m=n}^{\infty} \lambda^{m-n} \E  \Big[  U (\aaa_{m}; \s_m, \h_{<m})  \Big|  (\aaa_{n};\s_n,\h_{<n}) = (a_{n} ; s_n , h_{<n})  \Big]
\end{align*}
for all $(a_n;s_n,h_{<n})$ in $\A_n \times \SSS_n \times \H_n$. Using the bounds on the utility in (\ref{U_bounds}), it follows 
\begin{align*}
m_0 = (1-\lambda)\sum_{m=n}^{\infty} \lambda^{m-n} m_0 \leq \Phi_n (a_{n};s_n,h_{<n}) \leq (1-\lambda)\sum_{m=n}^{\infty} \lambda^{m-n} M_0 = M_0
\end{align*}
for all $(a_n;s_n,h_{<n})$.
\end{proof}

If we build the sequence $(\Phi_n)_{n \geq 0}$ using the recursive formula (\ref{RECURSION}), we must ensure it satisfies the above boundedness condition. If the autocrat intends to enforce a particular value, then there must exist a a behavioural strategy $\p$ which ensures the boundedness of $\Phi_n$ regardless of the sequence of moves employed by the adversary. This motivates the following definition.  
\begin{defn}
\label{DEFN_unenforceable}
Suppose the autocrat wishes to unilaterally enforce $\E \U_{\lambda} = v$. If, for any behavioural strategy $\p$, there exists a $\p$-accessible (see Definition \ref{p_accessible}) sequence of play $h_{\leq n}$ such that $\Phi_n(h_{\leq n})$ is less than $m_0$, we say that $v$ is \textbf{left unenforceable for state $\s$}. Similarly, if there exists a sequence of play  $h_{\leq n}$, such that $\Phi_{n}(h_{\leq n})$ is greater than $M_0$, we say $v$ is \textbf{right unenforceable for state $\s$}. 
\end{defn} 

It will be useful later to encapsulate exactly what sequence of opposing actions ensure that the $\Phi_n$'s eventually violate the boundedness condition. In state $s$, these are the \textbf{left and right unenforcing replies} to $x$, denoted $y^+_s(x)$ and $y^-_s(x)$. These are defined
\begin{align}
\label{y_extremal_left}
y^+_s(x) & \defeq \argmax_{y \in \Y_s} \big\{ \lambda m_{T (x,y;s)} + U_{\lambda} (x,y;s) \big\}
\end{align}
and 
\begin{align}
\label{y_extremal_right}
y^-_s(x) & \defeq  \argmin_{y \in \Y_s} \big\{ \lambda M_{T (x,y;s)} + U_{\lambda} (x,y;s) \big\} .
\end{align}
Note these replies are not guaranteed by any particular strategy of the adversary, deterministic or stochastic. In the simultaneous-play setting, it would require foreknowledge of the autocrat's action on the account of their opponent. Instead we must concern ourselves with plays by the opponent that are instances consistent with some behavioural strategy.   

Similarly, if the autocrat intends to exclude a particular value, then there must exist a a behavioural strategy $\p$ which ensures the unboundedness of $\Phi_n$ regardless of the sequence of moves employed by the adversary. This motivates the following definition. 

\begin{defn}
\label{DEFN_excludable}
Suppose the autocrat wishes to unilaterally exclude $\E \U_{\lambda}$ from attaining $v$. If there exists a strategy $\p$ such that every $\p$-accessible history of play $h_{\leq n}$ $\Phi_n(h_{\leq n})$ is eventually greater than $M_0$, we say that $v$ is \textbf{right excludable for state $\s$}. Similarly, if there exists $\p$ such that $\Phi_n$ is eventually less than $M_0$, we say that $v$ is \textbf{left excludable for state $\s$}.
 \end{defn}

\noindent It turns out the autocrat need only play a deterministic strategy to exclude a value. The autocrat's \textbf{right-excluding play} in $s$ is the extremal action $x^-_s$. Similarly, the \textbf{left-excluding play} $s$ is when the autocrat plays $x^+_s$ in state $s$. 

We now return to the sequences $(m_{s,n})_{n \geq 0}$ and $(\tilde{m}_{n \geq 0})$. The following two lemmas show they are respectively non-decreasing and non-increasing.

\begin{lem}
$m_{s,1} \geq m_0 \: \forall \: s$ and $\tilde{m}_{s,1} \leq \tilde{m}_0 \: \forall \: s$.
\end{lem}
\begin{proof}
We have
\begin{align*}
\lambda m_{T(x,y;s),0} + U_{\lambda} (x,y;s) = \lambda m_0 + (1-\lambda)U(x,y;s) \geq \lambda m_0 + (1-\lambda) m_0 = m_0 .
\end{align*}
This implies
\begin{align*}
\min_{x \in \X_s} \max_{y \in \Y_s} \big\{ \lambda m_{T(x,y;s),0} + U_{\lambda} (x,y;s)  \big\} \geq m_0 .
\end{align*}
The corresponding proof for $\tilde{m}_{s,1}$ is similar.
\end{proof}

\begin{lem}
If the truncated sequence $(m_{s,k})_{k \leq n}$ is non-decreasing, then $(m_{s,k})_{k \leq n+1} $ is non-decreasing. If the truncated sequence $(\tilde{m}_{s,k})_{k \leq n}$ is non-increasing, then $(\tilde{m}_{s,k})_{k \leq n+1} $ is non-increasing.
\end{lem}
\begin{proof}
\begin{align*}
 \lambda m_{T(x,y;s),n} + U_{\lambda} (x,y;s) & \geq \lambda m_{T(x,y;s),n-1} + U_{\lambda} (x,y;s)
\end{align*}
Taking the $\max$ over $Y_s$ implies
\begin{align*}
  \max_{y \in \Y_s} \big\{ \lambda m_{T(x,y;s),n} + U_{\lambda} (x,y;s) \big\} & \geq \max_{y \in \Y_s} \big\{ \lambda m_{T(x,y;s),n-1} + U_{\lambda} (x,y;s) \big\}.
\end{align*}
Finally, taking the $\min$ over $\X_s$ implies
\begin{align*}
\min_{x \in \X_s}\max_{y \in \Y_s} \big\{ \lambda m_{T(x,y;s),n} + U_{\lambda} (x,y;s) \big\}  & \geq \min_{x \in \X_s}\max_{y \in \Y_s} \big\{ \lambda m_{T(x,y;s),n-1} + U_{\lambda} (x,y;s) \big\}. 
\end{align*}
Thus $m_{s,n+1} \geq m_{s,n}$. The corresponding proof for $(\tilde{m}_{s,n})_{n \geq 0}$ is similar.
\end{proof}
\noindent Next we show these sequences are bounded.
\begin{lem}
The sequences $(m_{s,n})_{n \geq 0}$ and $(\tilde{m}_{s,n})_{n \geq 0}$ are bounded.
\end{lem}
\begin{proof}
every step involves a convex combination of the payoffs $U(x,y;s)$. That is,
\begin{align*}
m_{s,n} = \lambda m_{T ( x^-_n, y^+_n (x^-_n);s ),n} + (1-\lambda) U ( x^-_n, y^+_n (x^-_n), y_n;s) .
\end{align*}
Thus $m_{s,n} \in [m_0,M_0]$. The corresponding proof for $(\tilde{m}_{s,n})_{n \geq 0}$ is similar.
\end{proof}
\noindent As $(m_{s,n})_{n \geq 0}$ and $(\tilde{m}_{s,n})_{n \geq 0}$ are monotonic and bounded, the limits
\begin{align*}
m_s \defeq \lim_{n \rightarrow \infty} m_{s,n} \quad \mbox{and} \quad \tilde{m}_s \defeq \lim_{n \rightarrow \infty} \tilde{m}_{s,n}
\end{align*}
are  well defined. Taking the limit of the recursive relation (\ref{m_recursion}), we observe that $(m_s)_{s \in \SSS}$ and $(\tilde{m}_s)_{s \in \SSS}$  satisfy the same fixed point equation
\begin{align}
\label{m_FPE}
m_{s} \defeq \min_{x \in \X_s} \max_{y \in \Y_s} \big\{ \lambda m_{T(x,y;s)} + U_{\lambda} (x,y;s)  \big\} .
\end{align}
The above discussion the existence of solutions in Theorem \ref{THM_enforceable}. Uniqueness is covered by the following theorem.
\begin{theorem} (Uniqueness of the FPE)
\label{THM_uniqueness}
The solution to the fixed point equations in Theorem \ref{THM_enforceable} are unique.
\end{theorem}
\begin{proof}
If we can show $(m_{s,n})_{n \geq 0}$ equals $(\tilde{m}_{s,n})_{n \geq 0}$, then the monotonicity of the iterative process (\ref{m_recursion}) forbids the existence of any other fixed point.  

So suppose the solutions $(m_s)_{s \in \SSS}$ and $(\tilde{m}_s)_{s \in \SSS}$ to the (\ref{EQN_left_FPE}) are different. Suppose there exists a state $r$ such $m_r < \tilde{m}_r$. Lemma \ref{lem_rightmost} tells us that $m_r$ is left unenforceable.  Lemma \ref{LEM_left_un_right_ex} says that $m_r < m_r$, a contradiction. The case in which there exists a state $r$ such $m_r > \tilde{m}_r$ is handled similarly.
\end{proof}
\noindent Looking back at our definitions of left/right uneforceability and exludability, we have another immediate consequence stemming from uniqueness.
\begin{cor}
The supremum of left unenforceable values equals the minimum of right excludable values. The infimum of right unenforceable values equals the maximum of left excludable values.
\end{cor}

\noindent Finally, below are the lemmas used to prove Theorem \ref{THM_uniqueness}. Lemma \ref{LEM_left_un_right_ex} relies on lemmas \ref{seq_compactness} and \ref{reverse}.

\begin{lem}[Left unenforceable and right excludable values]
\label{LEM_left_un_right_ex}
$v < m_s$ iff $v$ is left unenforceable in state $s$. $v > \tilde{m}_s$ iff $v$ is right enforceable in state $s$.
\end{lem}
\begin{proof}
$(\Rightarrow)$ Suppose $v < m_s$. Take the smallest $n$ such that $v < m_{s,n+1}$. We claim the worst left reply eventually ensures the $\Phi_k$'s can take on values less than $m_0$. Since
\begin{align*}
\sum_{x \in \X_s} \Phi_k \left(x;s,h_{<k} \right)p_{x|s,h_{<k}} =  v ,
\end{align*}
there exists an $x_s$ such that $\Phi_k (x_s;s,h_{<k}) \leq m_{s,n+1} $. Thus
\begin{align*}
\Phi_k \left(x_s ;s,h_{<k} \right) < \min_{x \in \X_s} \max_{y \in \Y_s} \big\{ \lambda m_{T (x,y;s) ,n} + U_{\lambda} \left( x, y ;s \right) \big\} \leq \lambda m_{T (x_s , y_{x_s};s),n} + U_{\lambda} (x_s , y_{x_s};s)
\end{align*}
which implies
\begin{align*}
\lambda^{-1} \left[ \Phi_k (x_s  ;s,h_{<k}) - U( x_s ,y_{x_s};s) \right] < m_{T (x_s,y_{x_s};s),n} 
\end{align*}
So we have
\begin{align*}
\sum_{x' \in \X_s} \Phi_{k+1} \big( x' ; T ( x_s,y_{x_s};s ) ; h_{\leq k} \big) p_{x'| T(x_s,y_{x_s};s),h_{\leq k}} <  m_{T ( x_s,y_{x_s};s),n} 
\end{align*}
Repeat until some value of a future $\Phi_{k'}$ is less than $m_0$. \\

\noindent $(\Leftarrow)$ Suppose $v \geq m_s$. We will further assume that $v$ is left unenforceable. By combining Lemmas (\ref{seq_compactness}) and (\ref{reverse}) we have a contradiction.
\end{proof}

\begin{lem}
\label{seq_compactness} (Uniform left unenforceability and right excludability)
Suppose $v$ is left unenforceable in state $s_0$. There exists an $n \geq 0$ such that some value of $\Phi_n$ is less than $m_0$ for all $\p$. Suppose $v$ is right excludable in state $s$. If the autocrat uses the best right play $\bbb_r$, there exists an $n \geq 0$ such that all values of $\Phi_n [\bbb_r]$ are greater than $M_0$.
\end{lem}
\begin{proof}
Otherwise, for each $j \geq 0$ there is a $\p(j)$ such that $\Phi_{j} [\p(j)] \geq m_0$ for all $\Y$-replies. Since the space of behavioural strategies is sequentially compact, there exists a convergent subsequence $\left( \p(j_k) \right)_{k \geq 0}$ with limit $\p$ such that $\Phi_{j_k}[\p] \geq m_0$ for all $k$ and any $\Y$-reply. This implies $v$ is not left unenforceable, a contradiction. 
\end{proof}

\begin{lem}
\label{reverse}
Suppose $v$ is left unenforceable in state $s$. Let $n$ be the smallest natural number such that some value of $\Phi_n$ is less than $m_0$. Then $v < m_{s,n}$. Suppose $v$ is right excludable in state $s$. Let $n$ be the smallest natural number such that all values of $\Phi_n$ is greater than $M_0$. Then $v > \tilde{m}_{s,n}$. 
\end{lem}
\begin{proof}
\underline{Base Case}: ($n=1$)  for each deterministic strategy $\dd$ there exists a $\Y$-reply $\RRR_{\Y}(\dd)$ that is
\begin{align*}
\sum_x \Phi_0 \left(x; s \right) d_{x|s} = v
\end{align*}
for each $x$ there exists $y_x$ such that
\begin{align*}
\lambda^{-1} \left[ v- U_{\lambda} (x,y_x;s) \right] < m_0 = m_{T(x,y_x;s),0} 
\end{align*}
which implies
\begin{align*}
v < \lambda m_{0} + U_{\lambda} (x,y_x;s) \leq \lambda m_{0} + U_{\lambda} (x,y^+_0(x);s) \quad \forall x
\end{align*}
thus $v < m_{s,1}$ \\
\underline{Inductive Step}: Assume true for $n$, we will show its true for $n+1$
\begin{align*}
\sum_{x \in \X_s} \Phi_0 \left(x; s \right) d_{x|s} = v
\end{align*}
for each $x$ there exists $y_x$ such that
\begin{align*}
v_1 (x,y_x) \defeq \lambda^{-1} \left[ v - U_{\lambda} (x,y_x;s) \right] 
\end{align*}
is unenforceable in state $T (x,y_x;s)$. Applying the inductive assumption to $v_1 (x,y_x)$ in state $T (x,y_x;s)$ we conclude that $v_1(x,y_x) < m_{T(x,y_x;s),n}$ for every $x$. That is,
\begin{align*}
 \lambda^{-1} \left[ v - U_{\lambda} (x,y_x;s) \right] < m_{T(x,y_x;s),n} \quad \forall x
\end{align*}
which implies
\begin{align*}
v < \lambda m_{T(x,y_x;s),n} + U_{\lambda} (x,y_x;s) \leq \lambda m_{T(x,y^+_n(x);s),n} + U_{\lambda} (x,y^+_n(x);s) \quad \forall \: x
\end{align*}
Thus $v < m_{s,n+1}$.
\end{proof}
.
\begin{lem}
\label{lem_rightmost}
Let $(m_s)_{s \in \SSS}$ be a solution to the FPE (\ref{m_FPE}). If $v < m_s$ than $v$ is left unenforceable in state $s$. If $v > m_s$ than $v$ is right excludable in state $s$.
\end{lem}
\begin{proof}
Fix solution $(m_s)_{s \in \SSS}$ of the fixed point equation. Suppose $v < m_s$. We will show that $v$ is left unenforceable in state s. Suppose $x \in \X_s$ is countered with $y_x \in \Y_s$ in every $s$. Since $\SSS$ is finite, it will eventually enter a cycle of states, denoted $s_1 \rightarrow \ldots \rightarrow s_r \rightarrow s_1 $. Defining
\begin{align*}
U_i \defeq U_{\lambda} \left( x , y_x ; s_i \right),
\end{align*}
we have the action on $v$ after completing one circuit is
\begin{align*}
\lambda^{-r} \big[ v - \left(U_1 + \lambda U_2 + \ldots + \lambda^{r-1} U_r \right) \big]
\end{align*}
with fixed point $p \defeq (1-\lambda^r)^{-1} \left[ U_1 + \lambda U_{2} + \ldots + \lambda^{r-1} U_r \right]$.  We know
\begin{align*}
m_{s_i} - \lambda m_{s_{i+1}} \leq U_i
\end{align*}
which implies
\begin{align*}
\sum_{i=1}^r \lambda^{i-1} U_i \geq \sum_{i=1}^r \lambda^{i-1} \left( m_{s_i} - \lambda m_{s_{i+1}} \right) = \left( 1 - \lambda^r \right) m_{s_1}
\end{align*}
Thus $p \geq m_{s_1}$. Thus multiple circuits cause $v$ to tend towards $- \infty$ and thus it is left unenforceable.
\end{proof}

\subsection{Exact Formulas}

\noindent Equation (\ref{BIG}) implies the game graph can be decomposed into connected components with exactly one cycle each. For example, a cycle of of $n$ nodes $s_0, \ldots , s_n$, as in Figure \ref{DGRAM_cycle}, satisfy 
\begin{align}
\label{EQN_exact_cycle}
m_{s_i} = \lambda m_{s_{i+1}} +(1-\lambda)U(x_i,y_i;s_i)
\end{align}
for some $x_i$'s and $y_i$'s. This implies 
\begin{align}
m_{s_0} = \frac{\sum_{i=0}^{n-1} \lambda^i U( x_i,y_i;s_i) }{\sum_{i=0}^{n-1} \lambda^i}.
\end{align}
\noindent A connected component can have branches. Suppose the state $r_0$ is on a branch that connects to state $s_0$ of a cycle through the nodes $r_1 , \ldots , r_n$, as in Figure \ref{DGRAM_branch}. Once the extremal values of the cycle node $s_0$ is calculated, the corresponding values for the branch node $r_0$ is
\begin{align}
\label{EQN_exact_branch}
m_{r_0} = (1-\lambda)\sum_{i=0}^{n-1} \lambda^i U \left( x_i,y_i;r_i \right) + \lambda^n m_{s_0}.
\end{align}

\begin{figure}[h!]
    \begin{subfigure}{0.49\textwidth}
        \centering
        \begin{tikzpicture}[node distance={18mm}, thick,main/.style = {draw, circle}] 

\node[main] (1) {$s_0$}; 
\node[main] (2) [above right of=1] {$s_1$};
\node[main] (3) [below right of=2] {$s_2$}; 
\node[main] (4) [below of=3] {$s_3$};
\node[main] (5) [below of=1] {$s_n$};

\draw[->] (1) -- (2) node[midway,  above left ] (edge1) { $U_0$};
\draw[->] (2) -- (3) node[midway,  above right ] (edge2) { $U_1$};
\draw[->] (3) -- (4) node[midway,  right ] (edge3) { $U_2$};

\path (4) to coordinate[pos=0.2] (aux-1) coordinate[pos=0.7] (aux-2) (5);
\draw[->]  (4) -- (aux-1)
           (aux-2)-- (5);
\draw[dotted]    (aux-1) to  (aux-2);
\draw[->] (5) -- (1) node[midway,  left ] (edge4) { $U_n$};
\end{tikzpicture}
        \caption{A cycle. Here $U_i =  U(x_i,y_i;s_i)$.}
    \label{DGRAM_cycle}
    \end{subfigure}
    \begin{subfigure}{0.49\textwidth}
        \centering
        \begin{tikzpicture}[node distance={18mm}, thick,main/.style = {draw, circle}] 

\node[main] (1) {$r_0$}; 
\node[main] (2) [right of=1] {$r_1$};
\node[main] (3) [right of=2] {$r_n$}; 
\node[main] (4) [below right of=3] {$s_0$};

\draw[->] (1) -- (2) node[midway,  above] (edge1) { $U_0$};

\path (2) to coordinate[pos=0.2] (aux-1) coordinate[pos=0.7] (aux-2) (3);
\draw[->]  (2) -- (aux-1)
           (aux-2)-- (3);
\draw[dotted]    (aux-1) to  (aux-2);

\draw[->] (3) -- (4)node[midway,  above right] (edge1) { $U_n$};

\end{tikzpicture} 
        \caption{A branch. Here $U_i =  U(x_i,y_i;r_i)$.}
    \label{DGRAM_branch}
    \end{subfigure}
    \label{DGRAM_cycle_and_branch}
\end{figure}
\noindent Similar formulas hold for $M_s$.

\subsection{Examples}

\begin{exmp}
\label{EXMP_agnostic_pathological}
Figure \ref{DGRAM_patho_example} shows a game made up of a cycle of $n+1$ states. In state $0$ we have $\X = \{0,1 \}$ and $\Y= \{e \}$. For states $1$ to $n$ the reverse is true.
\begin{figure}[h!]
   \centering   
   \begin{tikzpicture}[node distance={20mm}, thick,main/.style = {draw, circle}] 

\node[main] (1) {$s_n$};
\node[main] (2) [right of=1] {$s_0$};
\node[main] (3) [below right of=2] {$s_{1}$};
\node[main] (4) [below left of=3] {$s_2$};
\node[main] (5) [left of=4] {$s_3$};

\draw[->] (1) to [bend left] node[pos=0.5, above] {$1_y$} (2);
\draw[->] (1) to [bend right] node[pos=0.5, below] {$0_y$} (2);

\draw[->] (2) to [bend left=20] node[pos=0.5, right] {$1_x$} (3);
\draw[->] (2) to [bend right=20] node[pos=0.5, left] {$0_x$} (3);

\draw[->] (3) to [bend left] node[pos=0.5, below right] {$1_y$} (4);
\draw[->] (3) to [bend right] node[pos=0.5, left] {$0_y$} (4);

\draw[->] (4) to [bend left] node[pos=0.5, below] {$1_y$} (5);
\draw[->] (4) to [bend right] node[pos=0.5, above] {$0_y$} (5);

\path (5) to [bend left=20] coordinate[pos=0.2] (aux-1) coordinate[pos=0.7] (aux-2) (1);
\draw[->]  (5) -- (aux-1)
           (aux-2)-- (1);
\draw[dotted]    (aux-1) [bend left=15] to  (aux-2);

\path (5) to [bend right=20] coordinate[pos=0.2] (aux-3) coordinate[pos=0.7] (aux-4) (1);
\draw[->]  (5) -- (aux-3)
           (aux-4)-- (1);
\draw[dotted]    (aux-3) [bend right=15] to  (aux-4);

\end{tikzpicture}
   \caption{A game with $n+1$ states whose autocratic strategies require $n$ memory.}
   \label{DGRAM_patho_example}
 \end{figure} 
For $0 \leq i \leq n$, define the $U_i$ and $V_i$ as
\begin{align*}
m_i = \max_y \{ \lambda m_{i-1} + U_{\lambda}(y;i) \} = \lambda m_{i-1} + \max_y U_{\lambda}(y;i) \defeq \lambda m_{i-1} + (1-\lambda)U_i 
\end{align*}
and
\begin{align*}
M_i = \min_y \{ \lambda M_{i-1} + U_{\lambda}(y;i) \} = \lambda M_{i-1} + \min_y U_{\lambda}(y;i) \defeq \lambda M_{i-1} + (1-\lambda)V_i .
\end{align*}
Similarly, define $U_0$ and $V_0$ as
\begin{align*}
m_0 = \min_x \{ \lambda m_{n} + U_{\lambda}(x;0) \} = \lambda m_{n} + \min_x U_{\lambda}(x;0) \defeq \lambda m_{n} + (1-\lambda)U_0 
\end{align*}
and
\begin{align*}
M_0 = \max_x \{ \lambda M_{n} + U_{\lambda}(x;0) \} = \lambda M_{n} + \max_x U_{\lambda}(x;0) \defeq \lambda M_{n} + (1-\lambda)V_0 .
\end{align*}
If we choose the $U_i$'s and $V_i$'s such that $U_0 < V_0$, $U_i > V_i$ for $1 \leq i \leq n$ and
\begin{align*}
U_0 + U_1 + \cdots + U_{n-1} + U_n < V_0 + V_1 + \cdots + V_{n-1} + V_n,
\end{align*}
then by using the equation (\ref{EQN_exact_cycle}) and it $M_s$ analogue we get
\begin{align*}
m_i = \frac{U_i+\lambda U_{i+1} + \cdots + \lambda^{n-1}U_{i-2} +\lambda^n U_{i-1}}{1+ \lambda + \cdots + \lambda^{n-1} + \lambda^n }
\end{align*}
and
\begin{align*}
M_i = \frac{V_i+\lambda V_{i+1} + \cdots + \lambda^{n-1}V_{i-2} +\lambda^n V_{i-1}}{1+ \lambda + \cdots + \lambda^{n-1} + \lambda^n} .
\end{align*}
Thus, for large $\lambda$ we have
\begin{align*}
m_i \approx \frac{1}{n+1} \left[U_0 + \ldots + U_n \right] <  \frac{1}{n+1} \left[V_0 + \ldots + V_n \right] \approx M_i \quad \forall \; 0 \leq i \leq n.
\end{align*}
\end{exmp}

\begin{exmp}
\label{2state_example}
Here is an example from \cite{Su25398}, a two state version of the donation game in which cooperating in state $H$ grants a larger benefit to one's opponent than in state $L$. Mutual cooperation in one round lets the players take part in the high benefit donation game. Otherwise the players play the low benefit game. See Table \ref{TABLE} and Figure \ref{DGRAM_2state_example}. Cooperation incurs the cost $c$ in both states.

\begin{table}[h!]
\centering
\subfloat[State H, larger benefit B]{
\begin{tabular}{cc|c|c|}
    & \multicolumn{1}{c}{} & \multicolumn{2}{c}{Opponent}\\
    & \multicolumn{1}{c}{} & \multicolumn{1}{c}{$1$}  & \multicolumn{1}{c}{$0$} \\\cline{3-4}
    \multirow{2}*{Autocrat}  & $1$ & $B-c,B-c$ & $-c,B$ \\\cline{3-4}
    & $0$ & $B,-c$ & $0,0$ \\\cline{3-4}
\end{tabular}
}
\qquad
\subfloat[State L, smaller benefit b]{
\begin{tabular}{cc|c|c|}
    & \multicolumn{1}{c}{} & \multicolumn{2}{c}{Opponent}\\
    & \multicolumn{1}{c}{} & \multicolumn{1}{c}{$1$}  & \multicolumn{1}{c}{$0$} \\\cline{3-4}
    \multirow{2}*{Autocrat}  & $1$ & $b-c,b-c$ & $-c,b$ \\\cline{3-4}
    & $0$ & $b,-c$ & $0,0$ \\\cline{3-4}
\end{tabular}
}
\caption{Payoff matrices of $2$-state donation game}
\label{TABLE}
\end{table}

As the number of states is small, the enforceable values can be solved directly. Both $m_H$ and $m_L$ must equal zero. If
\begin{align*}
\frac{c}{B-b}< \lambda < 1 ,
\end{align*}
then $M_L=b$ and $M_H = \lambda b + (1-\lambda)B$. This requires $ c < B-b$. Enforceable values for smaller $\lambda$ exist and in fact $M_H(\lambda)$ and $M_L(\lambda)$ have discontinuities at $\lambda = \frac{c}{B-b}$. In particular, if
\begin{align*}
  \frac{-b + \sqrt{b^2+4c(B-b)} }{2(B-b)} \leq \lambda \leq \frac{c}{B-b} 
\end{align*}
then $M_L = \lambda (B-c)  + (1-\lambda)(b-c)$ and $M_H=B-c$. This requires the additional constraint $B-2b \leq c$.

\begin{figure}[h!]
  \centering    
  \begin{tikzpicture}[node distance={22mm}, thick,main/.style = {draw, circle}]

 
 \clip (-2,-2) rectangle (5,2);

\node[main] (1) {$H$};
\node[main] (2) [right of=1]  {$L$};

\draw[->] (1) to [out=-135,in=135,looseness=8, right] node[pos=0.5] {$11$} (1);
\draw[->] (1) to [bend left =15, above] node[pos=0.5] {$00$} (2);
\draw[->] (1) to [bend left =45, above] node[pos=0.5] {$01$} (2) (2);
\draw[->] (1) to [bend left =90, above] node[pos=0.5] {$10$} (2) (2);
\draw[->] (2) to [out=35,in=-35,looseness=8, left] node[pos=0.5] {$00$} (2);
\draw[->] (2) to [out=50,in=-50,looseness=14, left] node[pos=0.5] {$01$} (2);
\draw[->] (2) to [out=60,in=-60,looseness=18, right] node[pos=0.5] {$10$} (2);
\draw[->] (2) to [bend left =35, above] node[pos=0.5] {$11$} (1);

\end{tikzpicture}
  \caption{2-state donation game from \cite{Su25398}}.
  \label{DGRAM_2state_example}
\end{figure}
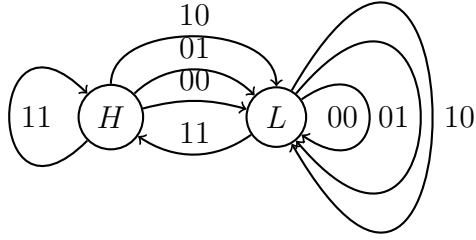

\end{exmp}

\section{The Algorithm}
\label{SECTION_algo}

We provide an algorithm for approximately solving the fixed point equations (\ref{BIG}). Once the extremal actions $x^+_s$ and $x^-_s$ are found (see (\ref{x_extremal_left}) and (\ref{x_extremal_right})), the exact formulas can be employed to check the validity of the results.
\begin{algo}
\label{ALGO}
\underline{Initialization}: $ \G(\SSS_0,\X_0 \times \Y,U,\lambda)$   \\

\underline{Loop}: 
\begin{enumerate}

\item 
Run the iterative process
\begin{align*}
m_{s,n+1} &= \min_{x \in \X_{i,s}} \max_{y \in \Y_{s}} \big\{ \lambda m_{T(x,y;s),n} + U_{\lambda} (x,y;s)  \big\} \\
M_{s,n+1} &= \max_{x \in \X_{i,s}} \min_{y \in \Y_{s}} \big\{ \lambda M_{T(x,y;s),n} + U_{\lambda} (x,y;s)  \big\}
\end{align*}
on the game graph $\G(\SSS_i,\X_i \times \Y,U,\lambda)$.

\item
For any edge $x \in \X_i$ that satisfies 
\begin{align*}
\max_{y \in \Y_s} \big\{ \lambda m_{T(x,y;s)} + U_{\lambda} (x,y;s)  \big\} \leq  \min_{y \in \Y_s} \big\{ \lambda M_{T \left( x,y;s \right)} + U_{\lambda} (x,y;s) \big\}, 
\end{align*} 
remove edges $(x,y)$ for each $y$ in $\Y_s$. If no extremal actions are removed in this manner then exit loop. Otherwise continue on to the graph pruning subloop.

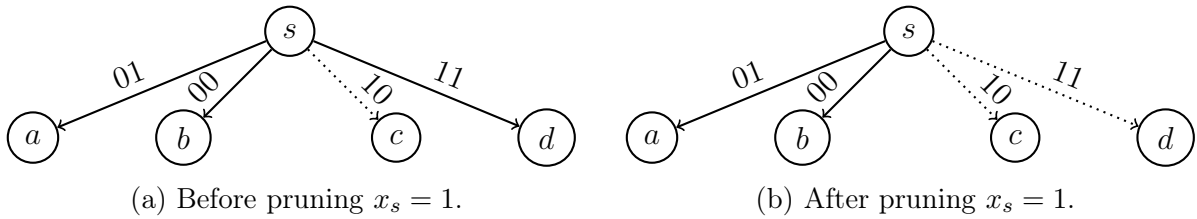
\begin{figure}[h!]
  \centering
 \begin{subfigure}[b]{0.49\textwidth}
   \centering   
   \begin{tikzpicture}[node distance={20mm}, thick,main/.style = {draw, circle}]

\tikzstyle{mystyle}=[draw,dashed,circle]; 

\node[main] (1) {$s$};
\node[main] (2) [below left  of=1 ] {$b$};
\node[main] (3) [left of=2] {$a$};
\node[main] (4) [below right of=1] {$c$};
\node[main] (5) [right of=4] {$d$};

\draw[->] (1) -- (2)  node[midway,  sloped , above left ] (edge1) { $00$} ;
\draw[->] (1) -- (3) node[midway,  sloped , above left ] (edge2) { $01$};
\draw[->,dotted] (1) -- (4) node[midway,  sloped , above right ] (edge3) { $10$} ;
\draw[->] (1) -- (5) node[midway,  sloped , above right] (edge4) { $11$};

\end{tikzpicture}
   \caption{Before pruning $x_s=1$.}
   \label{DGRAM_edge_prune_1}
 \end{subfigure}   
 \begin{subfigure}[b]{0.49\textwidth}
   \centering    
   \begin{tikzpicture}[node distance={20mm}, thick,main/.style = {draw, circle}]

\tikzstyle{mystyle}=[draw,dashed,circle]; 

\node[main] (1) {$s$};
\node[main] (2) [below left  of=1 ] {$b$};
\node[main] (3) [left of=2] {$a$};
\node[main] (4) [below right of=1] {$c$};
\node[main] (5) [right of=4] {$d$};

\draw[->] (1) -- (2)  node[midway,  sloped , above left ] (edge1) { $00$} ;
\draw[->] (1) -- (3) node[midway,  sloped , above left ] (edge2) { $01$};
\draw[->,dotted] (1) -- (4) node[midway,  sloped , above right ] (edge3) { $10$} ;
\draw[->,dotted] (1) -- (5) node[midway,  sloped , above right] (edge4) { $11$};

\end{tikzpicture}
   \caption{After pruning $x_s=1$.}
   \label{DGRAM_edge_prune_2}
 \end{subfigure}
 \caption{Edge pruning}
    \label{DGRAM_edge_pruning}
\end{figure}
 
\underline{Node Pruning Subloop}: 

\begin{enumerate}

\item
For any node $s \in \SSS_i$ such that $s$ has no outgoing edges, remove node $s$ and all incoming edges, as illustrated in Figure \ref{DGRAM_node_pruning}. If the subgraph is stable, exit the subloop.

\item
Prune the edges of $\G(\SSS_i,\X_i \times \Y,U,\lambda)$ (Definition \ref{prune}) as illustrated in Figure \ref{DGRAM_edge_pruning}. If the subgraph is stable, exit the subloop.
\end{enumerate}

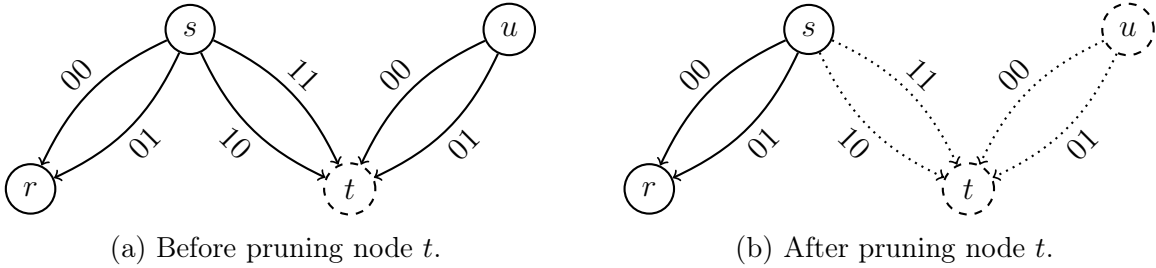
\begin{figure}[h!]
  \centering
 \begin{subfigure}[b]{0.49\textwidth}
   \centering   
   \begin{tikzpicture}[node distance={30mm}, thick,main/.style = {draw, circle}]

\tikzstyle{mystyle}=[draw,dashed,circle];

\node[main] (1) {$s$};
\node[main] (2) [below left of=1]  {$r$};
\node[mystyle] (3) [below right of=1]  {$t$};
\node[main] (4) [above right of=3]  {$u$};

\draw[->]  (1)  to [bend left=20,  below,sloped] node[pos=0.5] {$01$} (2) ;
\draw[->]  (1)  to [bend right=20,  above,sloped] node[pos=0.5] {$00$} (2) ;

\draw[->]  (1)  to [bend left=20,  above,sloped] node[pos=0.5] {$11$} (3) ;
\draw[->]  (1)  to [bend right=20, below,sloped] node[pos=0.5] {$10$} (3) ;

\draw[->]  (4)  to [bend left=20,  below,sloped] node[pos=0.5] {$01$} (3) ;
\draw[->]  (4)  to [bend right=20, above,sloped] node[pos=0.5] {$00$} (3) ;

\end{tikzpicture}
   \caption{Before pruning node $t$.}
   \label{DGRAM_node_prune_1}
 \end{subfigure}   
 \begin{subfigure}[b]{0.49\textwidth}
   \centering    
   \begin{tikzpicture}[node distance={30mm}, thick,main/.style = {draw, circle}]

\tikzstyle{mystyle}=[draw,dashed,circle];

\node[main] (1) {$s$};
\node[main] (2) [below left of=1]  {$r$};
\node[mystyle] (3) [below right of=1]  {$t$};
\node[mystyle] (4) [above right of=3]  {$u$};

\draw[->]  (1)  to [bend left=20,  below,sloped] node[pos=0.5] {$01$} (2) ;
\draw[->]  (1)  to [bend right=20,  above,sloped] node[pos=0.5] {$00$} (2) ;

\draw[->,dotted]  (1)  to [bend left=20,  above,sloped] node[pos=0.5] {$11$} (3) ;
\draw[->,dotted]  (1)  to [bend right=20, below,sloped] node[pos=0.5] {$10$} (3) ;

\draw[->,dotted]  (4)  to [bend left=20,  below,sloped] node[pos=0.5] {$01$} (3) ;
\draw[->,dotted]  (4)  to [bend right=20, above,sloped] node[pos=0.5] {$00$} (3) ;

\end{tikzpicture}
   \caption{After pruning node $t$.}
   \label{DGRAM_node_prune_2}
 \end{subfigure}
 \caption{Node pruning}
    \label{DGRAM_node_pruning}
\end{figure}
\end{enumerate}

\end{algo}

\noindent The global convergence of the iterative step follows from the monotonicity of the recursion.
\begin{lem}
The iterative processes
\begin{align*}
m_{s,n+1} \defeq \min_{x \in \X_s} \max_{y \in \Y_s} \big\{ \lambda m_{T(x,y;s),n} + U_{\lambda} (x,y,s)  \big\}
\end{align*}
and 
\begin{align*}
M_{s,n+1} \defeq \max_{x \in \X_s} \min_{y \in \Y_s} \big\{ \lambda M_{T(x,y;s),n} + U_{\lambda} (x,y,s)  \big\}
\end{align*}
are globally convergent.
\end{lem}
\begin{proof}
Suppose we initialize the process with $m'_{s,0} \in [m_0,M_0]$ for all $s$. Since the maximum and minimum functions are monotonically non-decreasing in each variable, we have
\begin{align*}
m_{s,1} \leq m'_{s,1} \leq \tilde{m}_{s,1} \quad \mbox{and} \quad  \tilde{M}_{s,1} \leq M'_{s,1} \leq M_{s,1}
\end{align*}
for all $s$ in $\SSS$. This is then true for all iterations, so taking the limit gives us  $m'_s = m_s$ and $M'_s = M_s$.
\end{proof}

After running the iterative process we prune the graph of any nodes where $m_s > M_s$. The iterative process is then rerun on the resulting subgraph. The following lemma shows how the extremal points of the subgraph compare to the initial graph.
\begin{lem}
If $\X_s' \subset X_s$ then $m_s \leq m'_s$ and $M'_s \leq M_s$. 
\end{lem}
\begin{proof}
We have
\begin{align*}
m_{s,1}= \min_{x \in X_s} \max_{y \in \Y_s} \big\{ \lambda m_{s,0} + U_{\lambda} (x,y,s)  \big\} \leq \min_{x \in X'_s} \max_{y \in \Y_s} \big\{ \lambda m_{s,0} + U_{\lambda} (x,y,s)  \big\} = m'_{s,1}.
\end{align*}
Iterating gives us $m_{s,n} \leq m'_{s,n}$ for all $s$ in $\SSS$ and $n \geq 0$. Taking the limit completes the proof. 
\end{proof}
\noindent Finally, we can estimate the worst case runtime of the globally convergent step.
\begin{lem}
For a tolerance of $\delta$, the iterative processes (\ref{m_recursion}) and (\ref{M_recursion}) have big $O$ of
\begin{align*}
 \max_s \{ |\X_s||\Y_s| \} \left| \SSS \right| \frac{  \ln \delta }{\ln \lambda}.
\end{align*}
\end{lem}
\begin{proof}
Fix an initialization $(m'_{s,0})_{s \in \SSS}$. Write each as $m'_{s,0} = m_s + \epsilon_{s,0}$. Then
\begin{align*}
m'_{s,1}= \min_{x \in X_s} \max_{y \in \Y_s} \big\{ \lambda (m_{s}+\epsilon_{s,0})  + U_{\lambda} (x,y,s)  \big\} \in  \lambda \left[ \min_s \epsilon_{s,0} , \max_s  \epsilon_{s,0} \right] + m_s.
\end{align*}
That is, $|m_s - m'_{s,1}| < \lambda |M_0 - m_0|$. Iterating this, we get
\begin{align*}
|m_s - m'_{s,n}| < \lambda^n |M_0 - m_0|
\end{align*}
for all $n \geq 0$. This means that with each iteration, the sequence gets closer to the fixed point by a factor of $\lambda$. From here we can calculate how many steps it takes to get within a tolerance $\delta$ of the true solution. That is,
\begin{align*}
\lambda^n |M_0 - m_0| < \delta  \quad \implies \quad n > \frac{1}{\ln \lambda} \ln \frac{\delta}{M_0 -m_0}.
\end{align*}
Noting that calculating the minimax takes $|\X_s||\Y_s|$ finishes the proof.
\end{proof}

\section*{Acknowledgements}
Thanks to Alex McAvoy and Christoph Hauert for their help in guiding me to this problem. For useful discussions and insight, thanks to Marcin P\k{e}ski and Adam Stinchecombe. Finally I am grateful to Almut Burchard for her advice, encouragement and unwavering support.

\bibliography{BIB}
\bibliographystyle{plain}

\end{document}